\documentclass{article}
\usepackage[asymmetric,left=20mm,right=20mm,bottom=36mm,top=33mm]{geometry}
\usepackage{amsthm}
\usepackage{logic9-thm}
\usepackage{tikz}
\usepackage{mathabx}
\usepackage{todonotes}
\usepackage{booktabs}

\theoremstyle{plain} 
\newtheorem{theorem}{Theorem}[section]
\newtheorem{lemma}[theorem]{Lemma}
\newtheorem{definition}[theorem]{Definition}

\newcommand{\bind}[1]{\downarrow_{#1}}
\newcommand{\val}{\nu}
\newcommand{\prvar}{\mathit{pr}}
\newcommand{\prvars}{\mathit{PR}}
\newcommand{\prvarl}{\mathit{pa}}
\newcommand{\ponet}{\mathit{pn?}}
\newcommand{\pol}{\mathit{po}}
\newcommand{\nel}{\mathit{ne}}
\newcommand{\sk}{\mathit{skip}}
\newcommand{\eval}{\mathit{eval}}
\newcommand{\fv}{\mathit{fv}}
\newcommand{\var}{\mathit{var}}
\newcommand{\fpt}{\text{\sc FPT}}
\newcommand{\wsat}{\text{\sc W[SAT]}}
\newcommand{\wop}{\text{\sc W[P]}}
\newcommand{\awsat}{\text{\sc AW[SAT]}}
\newcommand{\awp}{\text{\sc AW[P]}}

\newcommand{\xnl}{\text{\sc XNL}}
\newcommand{\aut}[1]{\mathcal{A}(#1)}
\newcommand{\xra}{\xrightarrow}
\newcommand{\mismatch}{\mathit{Mismatch}}

\newcommand{\range}{\mathit{Range}}

\newcommand{\Repeat}[2]{\noindent{\textsc{\textbf{#1~\ref{#2}.\
      }}}}

\title{Nesting Depth of Operators in Graph Database Queries: Expressiveness
Vs.~Evaluation Complexity}
\author{M. Praveen and B. Srivathsan\thanks{Both the authors are partially funded by
a grant from Infosys Foundation.}\\Chennai Mathematical Institute,
India}
\date{}

\begin{document}

\maketitle

\begin{abstract}
    Designing query languages for graph structured data is an active
    field of research, where expressiveness and efficient algorithms
    for query evaluation are conflicting goals. To better handle
    dynamically changing data, recent work has been done on designing
    query languages that can compare values stored in the graph
    database, without hard coding the values in the query. The main
    idea is to allow variables in the query and bind the variables to
    values when evaluating the query. For query languages that bind
    variables only once, query evaluation is usually \NP{}-complete.
    There are query languages that allow binding inside the scope of
    Kleene star operators, which can themselves be in the scope of
    bindings and so on. Uncontrolled nesting of binding and iteration
    within one another results in query evaluation being
    \PSPACE{}-complete.

    We define a way to syntactically control the nesting depth of
    iterated bindings, and study how this affects expressiveness and
    efficiency of query evaluation. The result is an infinite,
    syntactically defined hierarchy of expressions. We prove that the
    corresponding language hierarchy is strict. Given an expression in
    the hierarchy, we prove that it is undecidable to check if there is a
    language equivalent expression at lower levels. We prove that
    evaluating a query based on an expression at level $i$ can be done
    in $\S_{i}$ in the polynomial time hierarchy. Satisfiability of
    quantified Boolean formulas can be reduced to query evaluation; we study the
    relationship between alternations in Boolean quantifiers and the
    depth of nesting of iterated bindings.
\end{abstract}

\usetikzlibrary{fit}
\usetikzlibrary{shapes.geometric}

\newlength{\ml}
\setlength{\ml}{1cm}

\tikzstyle{state} = [circle, draw = black, minimum width = 0.2\ml,
inner sep=0.03\ml, fill=black]

\section{Introduction}
Graph structures representing data have found many applications like
semantic web \cite{GHM2011}, social networks \cite{RS2009} and
biological networks \cite{L2005}. Theoretical models of such data
typically have a graph with nodes representing entities and edges
representing relations among them. One reason for the popularity of
these models is their flexibility in handling semi-structured
data. While traditional relational databases impose rigid structures
on the relations between data elements, graph databases are better
equipped to handle data in which relations are not precisely known
and/or developing dynamically.

A fundamental query language for such models is Regular Path Queries
(RPQs), which is now part of the W3C recommendation \cite{SPARQL2013}.
An RPQ consists of a regular expression over the finite alphabet
labeling the edges of the graph.
Suppose a communication network is modeled by a graph, where nodes
represent servers and edges labeled $\ell$ represent links between
them. Evaluating the RPQ $\ell^{*}$ on this graph results in the set
of pairs of nodes between which there exists a route. Suppose each
link has a priority and we need pairs of connected nodes where all
intermediate links have the same priority. We can hard code the set of
priorities in the query. If the set of priorities is not static, a
querying mechanism which avoids hard coding is better. Every edge can
be labeled by a supplementary data value (priority of the link, in
this example) and we want query languages that can compare data values
without hard coding them in the syntax. Nodes can also carry data
values.  In generic frameworks, there is no a priori bound on the
number of possible data values and they are considered to be elements
of an infinite domain. Graph databases with data values are often
called data graphs in theory and property graphs in practice.

One way to design querying languages for data graphs is to extend RPQs
using frameworks that handle words on infinite alphabets
\cite{LV2012,LTV2013,KRV2014,Vrgoc2015}. Expressiveness and efficient
algorithms for query evaluation are conflicting goals for designing
such languages. We study a feature common to many of these languages,
and quantify how it affects the trade-off between expressiveness  and complexity of query
evaluation. Variable finite automata \cite{GKS2010} and parameterized
regular expressions \cite{BLR2013} are conservative extensions of
classical automata and regular expressions. They have variables, which
can be bound to letters of the alphabet at the beginning of query
evaluation. The query evaluation problem is \NP{}-complete for these
languages. Regular expressions with binding (REWBs) \cite{LTV2013} is
an extended formalism where binding of variables to values
can happen inside a Kleene star, which can itself be in the scope of another
binding operator and so on. Allowing binding and iteration to occur
inside each other's scope freely results in the query evaluation
problem being \PSPACE{}-complete. Here we study how the expressiveness
and complexity of query evaluation vary when we syntactically control
the depth of nesting of iterated bindings.

\emph{Contributions:}
\begin{enumerate}
    \item We syntactically classify REWBs according to the
        depth of nesting of iterated bindings.
    \item The resulting hierarchy of data languages is strict, and so
        is the expressiveness of queries.
    \item It is undecidable to check if a given REWB has a language
        equivalent one at lower levels.
    \item An REWB query in level $i$ can be evaluated in $\S_{i}$ in the
        polynomial time hierarchy.
    \item For lower bounds, we consider quantified Boolean formulas with some
        restrictions on quantifications and reduce their
        satisfiability to query evaluation, with some restrictions on
        the queries.
\end{enumerate}
For proving strictness of the language hierarchy, we build upon ideas
from the classic star height hierarchy \cite{Eggan1963}.
Universality of REWBs is known to be undecidable
\cite{NSV2004,KRV2014}.
We combine techniques from this proof with tools developed for the
language hierarchy to prove the third result above. The $\S_{i}$ upper
bound for query evaluation involves complexity theoretic arguments
based on the same tools. In the reductions from satisfiability of
quantified Boolean formulas to the query evaluation problem, the
relation between the number of alternations (in the Boolean
quantifiers) and the depth of nesting (of iterated bindings in REWBs)
is not straight forward. We examine this relation closely in the
framework of parameterized complexity theory, which is suitable for
studying the effect of varying the structure of input instances on the
complexity.

\emph{Related work:} The quest for efficient evaluation algorithms
and expressive languages to query graph databases, including those
with data values, is an active area of research; \cite{Barcelo2013} is
a recent comprehensive survey. Numerous formalisms based on logics and
automata exist for handling languages over infinite alphabets
\cite{Segoufin2006}. In \cite{LV2012}, the suitability of these
formalisms as query languages has been studied, zeroing in on register
automata mainly for reasons of efficient evaluation. The same paper
introduced regular expressions with memory and proved that they are
equivalent to register automata. REWBs \cite{LTV2013} have slightly
less expressive power but have better scoping structure for the
binding operator. Properties of these expressions have been further
studied in \cite{KRV2014}. In \cite{LMV2013}, XPath has been adapted
to query data graphs. Pebble automata have been adapted to work with
infinite alphabets in \cite{NSV2004}. A strict language hierarchy
based on the number of pebbles allowed in pebble automata has been
developed in \cite{Tan2013}. Many questions about comparative
expressiveness of register and pebble automata are open
\cite{NSV2004}. Fixed-point logics can be used to define languages over
infinite alphabets \cite{CM2014}. These logics can use the class
successor relation, which relates two positions with the same data
value if no intermediate position carries the same value.
Expressiveness of these logics increase \cite{CM2015, CM2015b}, when
the number of alternations between standard successor relation and class
successor relation increase.

\section{Preliminaries}
\label{sec:preliminaries}

\subsection{Data Languages and Querying Data Graphs}
We follow the notation of \cite{LTV2013}. Let $\S$ be a
finite alphabet and $\Dd$ a countably infinite set. The
elements of $\Dd$ are called \emph{data values}. A \emph{data word} is a
finite string over the alphabet $\Sigma \times \Dd$. We will write a
data word as $\binom{a_1}{d_1}\binom{a_2}{d_2}\dots \binom{a_n}{d_n}$,
where each $a_i \in \S$ and $d_i \in \Dd$. A set of data words is
called a \emph{data language}.

An extension of standard regular expressions, called \emph{regular
expressions with binding (REWB)}, has been defined in \cite{LTV2013}.
Here, data values are compared using variables.
For a set $\{x_1, x_2, \dots, x_k \}$ of variables, the set of
conditions $\Cc_k$ is the set of Boolean combinations of $x_{i}^{=}$
and $x_{i}^{\ne}$ for $i \in \set{1, \ldots, k}$. A data value $d \in
\Dd$ and a partial valuation $\val: \set{x_{1}, \ldots, x_{k}} \to
\Dd$ satisfies the condition $x_{i}^{=}$ (written as $d, \val \models
x_{i}^{=}$) if $\val(x_{i}) = d$. The satisfaction for other Boolean
operators are standard.

\begin{definition}[Regular expressions with binding (REWB)
    \cite{LTV2013}]
  Let $\S$ be a finite alphabet and $\{x_1,\dots, x_k\}$ a set of
  variables. Regular expressions with binding over $\S[x_1,\dots,x_k]$
  are defined inductively as: $~~r ~:=~ \e~|~a~|~a[c]~|~r + r~|~r \cdot
  r~|~r^*~|~a \bind{x}(r)~~$ 
  where $a\in \S$ is a letter in the alphabet, $c \in \Cc_k$ is a
  condition on the variables and $x \in \set{x_{1}, \ldots,
  x_{k}}$ is a variable.
\end{definition}

We call $\bind{x}$ the \emph{binding operator}. In the expression $a
\bind{x} (r)$, the expression $r$ is said to be the \emph{scope} of
the binding $\bind{x}$. A variable $x$ in an expression is
\emph{bound} if it occurs in the scope of a binding
$\bind{x}$. Otherwise it is \emph{free}.
We write $\fv(r)$ to denote the
set of free variables in $r$ and $r(\bar{x})$ to denote that
$\bar{x}$ is the sequence of all free variables.
The semantics of an REWB $r(\bar{x})$ over the variables $\{x_1,
\dots, x_k\}$ is defined with respect to a partial valuation
$\val: \{x_1, \dots, x_k\} \to \Dd$ of the variables. A valuation
$\val$ is \emph{compatible} with $r(\bar{x})$ if $\val(\bar{x})$ is
defined.


\begin{definition}[Semantics of REWB]
  Let $r(\bar{x})$ be an REWB over $\Sigma[x_1, \dots, x_k]$ and let
  $\val: \{x_1, \dots, x_k\} \to \Dd$ be a valuation of variables
  compatible with $r(\bar{x})$. The language of data words
  $L(r, \val)$ defined by $r(\bar{x})$ with respect to $\val$ is given
  as follows:

  \hfill\break
      \begin{tabular}[h]{rl|rl|rl}
          \toprule
          $r$ & $L(r,\val)$ & $r$ & $L(r,\val)$ & $r$ & $L(r.\val)$\\
          \midrule
          $\e$ & $\set{\e}$ & $a$ & $\{ \binom{a}{d}~|~ d \in \Dd \}$
          & $a[c]$ & $\{ \binom{a}{d}~|~ d, \val \models c \}$\\
          $r_{1} + r_{2}$ & $L(r_1, \val) \cup L(r_2,\val)$ &
          $r_{1} \cdot r_{2}$ & $L(r_1, \val) \cdot L(r_2, \val)$ &
          $r_{1}^{*}$ & $(L(r_{1}, \val))^{*}$\\
          $a \bind{x_i}(r_1)$ & $\bigcup_{d \in
      \Dd} \{\binom{a}{d}\} \cdot L(r_1, \val[x_i \to d])$ &
      \hphantom{a}&\\
      \bottomrule
      \end{tabular}\hfill\break
  where $\val[x_i \to d]$ denotes the valuation which is the
  same as $\val$ except for $x_i$ which is mapped to $d$. An REWB $r$
  defines the data language $L(r) = \bigcup_{\val \text{ compatible
  with } r} L(r, \val)$.
\end{definition}

For example, the REWB $a\bind{x} (b[x^=]^*)$ defines the set of data
words of the form $a b^*$ with all positions having the
same data value. The REWB $(a \bind{x} (b [x^=]))^*$ defines the set
of data words of the form
$\binom{a}{d_1}\binom{b}{d_1}\binom{a}{d_2}\binom{b}{d_2} \cdots
\binom{a}{d_n}\binom{b}{d_n}$.

%
\begin{definition}[Data graphs]
  A \emph{data graph} $G$ over a finite alphabet $\Sigma$ and an
  infinite set of data values $\Dd$ is a pair $(V, E)$ where $V$ is a
  finite set of vertices, and $E \incl V \times \Sigma \times \Dd \times V$
  is a set of edges which carry labels from $\Sigma \times \Dd$.
\end{definition}

We do not have data values on vertices, but
they can be introduced without affecting the results. A \emph{regular data path query} is of the form $Q=x \xra{r} y$ where
$r$ is an REWB. Evaluating $Q$ on a data graph $G$ results in the set $Q(G)$
of pairs of nodes $\struct{u,v}$ such that there exists a data path from
$u$ to $v$ and the sequence of labels along the data path forms a data word in
$L(r)$. Evaluating a regular data path query on a data
graph is known to be \PSPACE{}-complete in general and \NLOGSPACE{}-complete when
the query is of constant size \cite{LTV2013}.
We sometimes identify the query $Q$ with the expression
$r$ and write $r(G)$ for $Q(G)$. A query $r_{1}$ is said to be
contained in another query $r_{2}$ if for every data graph $G$,
$r_{1}(G) \subseteq r_{2}(G)$. It is known from \cite[Proposition
3.5]{KRV2014} that a query $r_{1}$ is contained in the query $r_{2}$
iff $L(r_{1}) \subseteq L(r_{2})$. Hence, if a class $E_{2}$ of REWBs
is more expressive than the class $E_{1}$ in terms of defining data
languages, $E_{2}$ can also express more queries
than $E_{1}$.
 
\subsection{Parameterized Complexity}
The size of queries are typically small compared to the size of
databases. To analyze the efficiency of query evaluation algorithms,
the size of the input can be naturally split into the size of the
query and the size of the database. Parameterized complexity theory is
a formal framework for dealing with such problems. An instance of a
parameterized problem is a pair $(x,k)$, where $x$ is an encoding of
the input structure on which the problem has to be solved (e.g., a
data graph and a query), and $k$ is a parameter associated with the input
(e.g., the size of the query). A parameterized problem is said to be
in the parameterized complexity class Fixed Parameter Tractable
(\fpt{}) if there is a computable function $f:\Nat \to \Nat$, a
constant $c \in \Nat$ and an algorithm to solve the problem in time
$f(k)|x|^{c}$.

We will see later that the query evaluation problem is unlikely to be
in \fpt{}, when parameterized by the size of the regular data path
query. There are many parameterized complexity classes that are
unlikely to be in \fpt{}, like \wsat{}, \wop{},
\awsat{} and \awp{}. To place parameterized problems in these classes,
we use \fpt{}-reductions.
\begin{definition}[\fpt{} reductions]
    \label{def:fptReduction}
    A \fpt{} reduction from a parameterized problem $Q$ to another
    parameterized problem $Q'$ is a mapping $R$ such that:
    \begin{enumerate}
        \item For all instances $(x,k)$ of parameterized problems,
            $(x,k) \in Q$ iff $R(x,k) \in Q'$.
        \item There exists a computable function $g:\Nat \to \Nat$
            such that for all $(x,k)$, say with $R(x,k) = (x',k')$, we
            have $k' \le g(k)$.
        \item There exist a computable function $f:\Nat \to \Nat$ and
            a constant $c \in \Nat$ such that $R$ is computable in
            time $f(k)|x|^{c}$.
    \end{enumerate}
\end{definition}


\section{Nesting Depth of Iterated Bindings and Expressive Power}
\label{sec:nest-depth-iter-short}


A binding $\bind{x}$ along with a condition $[x^=]$ or $[x^\neq]$ is
used to constrain the possible data values that can occur at certain
positions in a data word. A binding inside a star --- an
\emph{iterated binding} --- imposes the constraint arbitrarily many
times. For instance, the expression $r_1 := (a_1 \bind{x_1} (b_1
[x_1^=]))^*$ defines data words in $(a_1b_1)^*$ where every $a_1$ has
the same data value as the next $b_{1}$. We now define a syntactic
mechanism for controlling the nesting depth of iterated bindings. The
restrictions result in an infinite hierarchy of expressions. The
expressions at level $i$ are generated by $F_{i}$ in the grammar
below, defined by induction on $i$.
\begin{align*}
  F_{0} & ::= \e ~|~ a ~|~ a[c] ~|~ F_{0} + F_{0} ~|~ F_{0} \cdot
  F_{0} ~|~ F_{0}^{*}\\
  E_{i} & ::= F_{i-1} ~|~ E_{i} + E_{i} ~|~ E_{i} \cdot E_{i} ~|~ a
  \downarrow_{x_{j}} (E_{i})\\
  F_{i} & ::= E_{i} ~|~ F_{i} + F_{i} ~|~ F_{i} \cdot F_{i} ~|~
  F_{i}^{*}
\end{align*}
where $i \ge 1$, $a \in \S$, $c$ is a condition in $\mathcal{C}_{k}$
and $x_{j} \in \set{x_{1}, \ldots, x_{k}}$. Intuitively, $E_{i}$ can
add bindings over iterations (occurring in $F_{i-1}$) and $F_{i}$ can
add iterations over bindings (occurring in $E_{i}$). The nesting depth
of iterated bindings in an expression in $F_{i}$ is therefore
$i$. The union of all expressions at all levels equals the set of
REWBs. In this paper, we use subscripts to denote the levels of
expressions and superscripts to denote different expressions in a
level: so $e_{5}^{1}$ is some expression in $E_{5}$, $f_{3}^{2}$ is
some expression in $F_{3}$.

We now give a sequence of expressions $\{r_i\}_{i \ge 1}$ such that
each $r_i$ is in $F_{i}$ but no language equivalent expression exists
in $F_{i-1}$. For technical convenience, we use an unbounded number of
letters from the finite alphabet and an unbounded set of
variables. The results can be obtained with a constant number of
letters and variables.
\begin{definition}
  \label{def:strict-expressions}
  Let $\{a_1, b_1, a_2, b_2, \dots \}$ be an alphabet and $\{x_1, x_2,
  \dots\}$ a set of variables. We define $r_1$ to be $(a_1 \bind{x_1} (b_1
  [x_1^=]))^*$. For $i \ge 2$, define $r_i := (a_i \bind{x_i} (r_{i-1}
  b_i[x_i^=]))^*$.
\end{definition}

From the syntax, it can be seen that each $r_i$ is in $F_i$. To show
that $L(r_i)$ cannot be defined by any expression in $F_{i-1}$, we
will use an ``automaton view'' of the expression, as this makes
pigeon-hole arguments simpler. No automata characterizations are known
for REWBs in general; the restrictions on the binding and star
operators in the expressions of a given level help us build specific
automata in stages.

Standard finite state automata can be converted to regular
expressions by considering generalized non-deterministic finite
automata, where transitions are labeled with regular expressions
instead of a single letter (see e.g., \cite[Lemma 1.32]{SIP1997}). The
language of an expression $f_{i}^{1}$ can be accepted by such an
automaton, where transitions are labeled with expressions in $E_{i}$.
We will denote this automaton by $\aut{f_i^1}$. Similarly, the
language of an expression $e_{i}^{1}$ can be accepted by an automaton
whose transitions are labeled with expressions in $F_{i-1}$ or with $a
\bind{x}$. We will denote this automaton by $\aut{e_{i}^{1}}$.  There
are no cycles in $\aut{e_{i}^{1}}$, since $e_{i}^{1}$ can not use the
Kleene $*$ operator except inside expressions in $F_{i-1}$. The runs
of $\aut{e_{i}^{1}}$ are sequences of pairs of a state and a valuation
for variables. The valuations are updated after every transition with
a label of the form $a\bind{x}$. Formal semantics are given in
Appendix~\ref{sec:append-sect-hierarchy}, which also contains all the proofs in detail.

We will prove that $L(r_i)$ cannot be defined by any expression in
$E_{i}$ (and hence not by any expression in $F_{i-1}$). We first
define the following sequence of words, which will be used in the
proof. Let $\{d[j_{1}, j_{2}] \in \Dd \mid j_{1}, j_{2} \in \Nat \}$
be a set of data values such that $d[j_{1},j_{2}] \ne
d[j_{1}',j_{2}']$ if $\di{j_{1},j_{2}} \ne \di{j_{1}',j_{2}'}$. For
every $n \ge 1$, define the words:
\begin{align*}
  u_{1,n} &:=
  \binom{a_1}{d[1,1]}\binom{b_1}{d[1,1]}\binom{a_1}{d[1,2]}\binom{b_1}{d[1,2]}~\cdots~\binom{a_1}{d[1,n^{2}]}
  \binom{b_1}{d[1,n^{2}]}
  \\
  u_{i,n} & := \binom{a_i}{d[i,1]}u_{i-1,n}\binom{b_i}{d[i,1]}\binom{a_i}{d[i,2]}u_{i-1,n}\binom{b_i}{d[i,2]}\cdots
  \binom{a_i}{d[i,n^2]}u_{i-1,n}\binom{b_i}{d[i,n^2]} \\ &\text{
    for all } i \ge 2
\end{align*}

In order to prove that $L(r_i)$ cannot be defined by any expression in
$E_{i}$, we will show the following property: if $u_{i,n}$ occurs as a
sub-word of a word $w$ in the language of a ``sufficiently small''
expression $e_{i}^1$, then the same expression accepts a word where
some $a_j$ and a matching $b_j$ have different data values. Let
$\mismatch_{i,n}$ be the set of all data words obtained from $u_{i,n}$
by modifying the data values so that there exist two positions $p, p'$
with $p < p'$ and a $j \le i$ such that: $p$ contains $\binom{a_j}{d}$
and $p'$ contains $\binom{b_j}{d'}$ with $d \neq d'$; moreover between
positions $p$ and $p'$, $b_j$ does not occur in the word. We consider
expressions in which no two occurrences of the binding operator use
the same variable. For an expression $e$, let $|\aut{e}|$ denote the
number of states in the automaton $\aut{e}$ and $|\var(e)|$ denote the
number of variables in $e$.

\begin{lemma}
  \label{lem:u-words-property}
  Let $e_{i}^1$ be an expression and let $n \in \Nat$ be greater than
  $(|\aut{e}|+1)$ and $(|\var(e)|+1)$ for every sub-expression $e$ of
  $e_i^1$. Let $\val$ be a valuation of $\fv(e_{i}^1)$ and let $x, z$
  be data words. Then: $xu_{i,n} z \in L(e_{i}^1, \val) ~~\implies~~
  x\bar{u}_{i,n} z \in L(e_{i}^1, \val) \text{ for some }
  \bar{u}_{i,n} \in \mismatch_{i,n}$.
\end{lemma}
\begin{proof}[Proof idea] 
  By induction on $i$. Suppose $x u_{i,n} z \in L(e_i^1, \nu)$. 
  The run of $\aut{e_{i}^{1}}$ on $x u_{i,n} z$ consists
  of at most $n$ transitions, since the automaton is acyclic and has
  at most $n$ states. Each of the (at most) $n$ transitions reads some
  sub-word in the language of some sub-expression $f_{i-1}^{1}$, while
  the whole word consists of $n^{2}$ occurrences of
  $a_{i}u_{i-1,n}b_{i}$.  Hence, at least one sub word consists of
  $n$ occurrences of $a_{i}u_{i-1,n}b_{i}$. A run of
  $\aut{f_{i-1}^{1}}$ on such a sub-word is shown below.
  
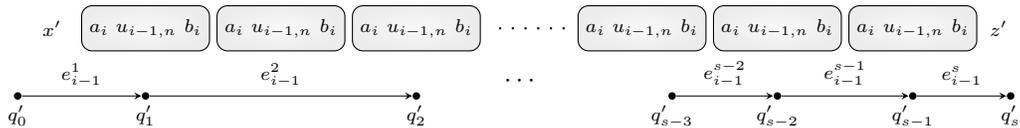
\begin{center}
    \begin{tikzpicture}
      \node [left] at (-1, 0) {\scriptsize $x'$}; \node [right] at
      (11.1, 0) {\scriptsize $z'$};

\begin{scope}
  \shade [rounded corners, inner color=gray!5, outer color=gray!15,
  draw=black] (-0.85, -0.3) rectangle (0.85, 0.3); \node at (0,0)
  {\scriptsize $a_i ~u_{i-1,n}~ b_i$};
\end{scope}

\begin{scope}[xshift=1.8cm]
  \shade [rounded corners, inner color=gray!5, outer color=gray!15,
  draw=black] (-0.85, -0.3) rectangle (0.85, 0.3); \node at (0,0)
  {\scriptsize $a_i ~u_{i-1,n}~ b_i$};
\end{scope}

\begin{scope}[xshift=3.6cm]
  \shade [rounded corners, inner color=gray!5, outer color=gray!15,
  draw=black] (-0.85, -0.3) rectangle (0.85, 0.3); \node at (0,0)
  {\scriptsize $a_i ~u_{i-1,n}~ b_i$};
\end{scope}

\begin{scope}[xshift=6.6cm]
  \shade [rounded corners, inner color=gray!5, outer color=gray!15,
  draw=black] (-0.85, -0.3) rectangle (0.85, 0.3); \node at (0,0)
  {\scriptsize $a_i ~u_{i-1,n}~ b_i$};
\end{scope}

\begin{scope}[xshift=8.4cm]
  \shade [rounded corners, inner color=gray!5, outer color=gray!15,
  draw=black] (-0.85, -0.3) rectangle (0.85, 0.3); \node at (0,0)
  {\scriptsize $a_i ~u_{i-1,n}~ b_i$};
\end{scope}

\begin{scope}[xshift=10.2cm]
  \shade [rounded corners, inner color=gray!5, outer color=gray!15,
  draw=black] (-0.85, -0.3) rectangle (0.85, 0.3); \node at (0,0)
  {\scriptsize $a_i ~u_{i-1,n}~ b_i$};
\end{scope}

\draw [thick, loosely dotted] (4.7, 0) -- (5.7,0);

\node [fill, circle, inner sep=1pt] (0) at (-1.7, -0.9) {}; \node
[fill, circle, inner sep=1pt] (1) at (0, -0.9) {}; \node [fill,
circle, inner sep=1pt] (2) at (3.6, -0.9) {}; \node [fill, circle,
inner sep=1pt] (s1) at (7, -0.9) {}; \node [fill, circle, inner
sep=1pt] (s2) at (8.4, -0.9) {}; \node [fill, circle, inner sep=1pt]
(s3) at (10.2, -0.9) {}; \node [fill, circle, inner sep=1pt] (s4) at
(11.5, -0.9) {};

\node [below] at (-1.7, -0.9) {\scriptsize $q'_0$}; \node [below] at
(-0, -0.9) {\scriptsize $q'_1$}; \node [below] at (3.6, -0.9)
{\scriptsize $q'_2$}; \node [below] at (7, -0.9) {\scriptsize
  $q'_{s-3}$}; \node [below] at (8.4, -0.9) {\scriptsize $q'_{s-2}$};
\node [below] at (10.2, -0.9) {\scriptsize $q'_{s-1}$}; \node [below]
at (11.5, -0.9) {\scriptsize $q'_s$};

\begin{scope}[very thin, ->, >=stealth]
    \draw (0) -- node[auto=left] {\scriptsize $e_{i-1}^{1}$} (1); \draw (1) --
    node[auto=left] {\scriptsize $e_{i-1}^{2}$} (2); \draw (s1) --
    node[auto=left] {\scriptsize $e_{i-1}^{s-2}$} (s2); \draw (s2)
  -- node[auto=left] {\scriptsize $e_{i-1}^{s-1}$} (s3); \draw (s3) --
  node[auto=left] {\scriptsize $e_{i-1}^{s}$} (s4);
\end{scope}

\node at (5, -0.7) {$\cdots$};
\end{tikzpicture}
\end{center}
Every transition of this run reads sub-words in
the language of some sub-expression $e_{i-1}^{j}$. If some transition
of this run reads an entire sub-word $u_{i-1,n}$ (as in transition
$q'_1 \xra{~~} q'_2$), then we can create a mismatch inside this
$u_{i-1,n}$ by induction hypothesis. Otherwise, none of the
transitions read an $a_{i}$ and the corresponding $b_{i}$ together (as
in $q'_{s-2} \xra{~~} q'_{s-1}$ in the figure). None of the $b_{i}$s
is compared with the corresponding $a_{i}$, so the data value of one
of the $b_{i}$s can be changed to create a mismatch. The resulting
data word will be accepted provided the change does not result in a
violation of some condition. Since the range of the valuation has at
most $(n-1)$ distinct values, one of the $n$ $b_{i}$s is safe for
changing the data value.
\end{proof}

\begin{theorem}
  For any $i$, the language $L(r_i)$ cannot be defined by any
  expression in $E_{i}$.
\end{theorem}
\begin{proof}
  Suppose $r_i$ is equivalent to an expression $e_{i}^1$. Pick an $n$
  bigger than $|\aut{e}|$ and $|\fv(e)|$ for every sub-expression $e$
  of $e_i^{1}$. The word $u_{i,n}$ belongs to $L(r_i)$ and hence
  $L(e_{i}^1)$. By Lemma~\ref{lem:u-words-property}, we know that if
  this is the case, then $\bar{u}_{i,n} \in L(r_i)$ for some word
  $\bar{u}_{i,n} \in \mismatch_{i,n}$. But $L(r_i)$ cannot contain
  words with a mismatch. A contradiction.
\end{proof}

Given an expression at some level, it is possible that its language is
defined by an expression at lower levels. Next we show that it is
undecidable to check this.

\begin{theorem}
  \label{thm:undec-height}
  Given an expression in $F_{i+1}$, checking if there exists a
  language equivalent expression in $F_i$ is undecidable.
\end{theorem}
\begin{proof}[Proof idea] By reduction from Post's Correspondence
    Problem (PCP). The basic idea is from the proof of undecidability
    of universality of REWBs and related formalisms
    \cite{NSV2004,LTV2013}.  For an instance $\set{(u_{1},v_{1}),
    \ldots, (u_{n},v_{n})}$ of PCP, a solution (if it exists) can be
    encoded by a data word of the form $w_{1} \#r_{i} \# w_{2}$, where
    $w_{1}$ is made up of $u_{i}$'s, $w_{2}$ is made up of $v_{i}$'s
    and $r_{i}$ is from Definition~\ref{def:strict-expressions}. To
    ensure that such a data word indeed represents a solution, we need
    to match up the $u_{i}$'s in $w_{1}$ with the $v_{i}$'s in $w_{2}$,
    which can be done through matching data values.
    Consider the language of data words of the form $w_{1}' \# r_{i}
    \# w_{2}'$ that are \emph{not} solutions of the given PCP
    instance.  This language can be defined by an expression $\D$ in
    $E_{i+1}$, which compares data values in the left of $\# r_{i} \#$
    with those on the right side, to catch mismatches. We can prove
    that no equivalent expression exists in lower levels, using
    techniques used in Lemma~\ref{lem:u-words-property}. On the other
    hand, if the given PCP instance does not have a solution, no data
    word encodes a solution, so the given language is defined by
    $\S^{*}r_{i}\S^{*}$, which is in $F_{i}$.
\end{proof}


\section{Complexity of Query Evaluation}
\label{sec:query-evaluation}
In this section, we will study how the depth of nesting of iterated
bindings affects the complexity of evaluating queries. An instance of
the query evaluation problem consists of a data graph $G$, an REWB
$e$, a valuation $\val$ for $\fv(e)$ and a pair $\struct{u,v}$ of
nodes in $G$. The goal is to check if $u$ is connected to $v$ by a
data path in $L(e,\val)$.

\subsection{Upper Bounds}
An expression in $F_{i}$ can be thought of as a standard regular
expression (without data values) over the alphabet of its
sub-expressions. This is the main idea behind our upper bound results.
The main result proves that evaluating queries in $E_{i}$ can be done
in $\S_{i}$ in the polynomial time hierarchy.
\begin{lemma}
  \label{lem:evalFiWithOracleForEi}
  With an oracle for evaluating $E_{i}$ queries,
  $F_{i}$ queries can be evaluated in polynomial time.
\end{lemma}
\begin{proof}[Proof idea]
  Suppose the query $f_{i}^{1}$ is to be evaluated on the data graph
  $G$ and $f_{i}^{1}$ consists of the sub-expressions $e_{i}^{1},
  \ldots, e_{i}^{m}$ in $E_{i}$.  For every $j$, add an edge labeled
  $e_{i}^{j}$ between those pairs $\struct{v_{1},v_{2}}$ of nodes of
  $G$ for which $\struct{v_{1}, v_{2}}$ is in the evaluation of
  $e_{i}^{j}$ on $G$. Evaluating the sub-expressions can be done with
  the oracle. Now $f_{i}^{1}$ can be treated as a standard regular
  expression over the finite alphabet $\set{e_{i}^{1}, \ldots,
  e_{i}^{m}}$,  and can be evaluated in polynomial time using standard
  automata theoretic techniques.
\end{proof}
\begin{theorem}
  \label{thm:complexityEvaluation}
  For queries in $E_{i}$, the evaluation problem belongs to $\S_{i}$.
\end{theorem}
\begin{proof}[Proof idea]
    Since bindings in $E_{i}$ are not iterated, each binding is
    performed at most once. The data value for each variable is
    guessed non-deterministically. The expression can be treated as a
    standard regular expression over its sub-expressions and the
    guessed data values. The sub-expressions are in $F_{i-1}$, which
    can be evaluated in polynomial time
    (Lemma~\ref{lem:evalFiWithOracleForEi}) with an oracle for
    evaluating queries in $E_{i-1}$. This argument will not work in
    general for arbitrary REWBs --- bindings that are nested deeply
    inside iterations and other bindings may occur more than
    polynomially many times in a single path.
\end{proof}

Next we consider the query evaluation problem with the size of the
query as the parameter. An instance of the \emph{parameterized
weighted circuit satisfiability} problem consists of a Boolean circuit
and the parameter $k \in \Nat$. The goal is to check if the circuit
can be satisfied by a truth assignment of weight $k$ (i.e., one
that sets exactly $k$ propositional atoms to true). The
class \wop{} is the set of all parameterized problems which are
\fpt{}-reducible to the weighted circuit satisfiability problem.
\begin{theorem}
    \label{thm:evalLevelInWP}
    Evaluating REWB queries in $E_{1}$, parameterized by the size of
    the query is in \wop{}.
\end{theorem}
\begin{proof}[Proof idea]
    It is proved in \cite[Lemma~7, Theorem~8]{CFG2003} that a
    parameterized problem is in \wop{} iff there is a
    non-deterministic Turing machine that takes an instance $(x,k)$
    and decides the answer within $f(k)|x|^{c}$ steps, of which at
    most $f(k) \log |x|$ are non-deterministic (for some computable
    function $f$ and a constant $c$). Such a Turing machine exists for
    evaluating REWB queries in $E_{1}$, using the steps outlined in
    the proof idea of Theorem~\ref{thm:complexityEvaluation}.
\end{proof}
Thus, the number of non-deterministic steps needed to evaluate an
$E_{1}$ query depends only logarithmically on the size of the data
graph. It is also known that \wop{} is contained in the class {\sc
para-NP} --- the class of parameterized problems for which there are
deterministic algorithms taking instances $(x,k)$ and computing an
equivalent instance of the Boolean satisfiability problem in time
$f(k)|x|^{c}$. Hence, we can get an efficient reduction to the
satisfiability problem, on which state of the art \textsc{sat} solvers
can be run. Many hard problems in planning fall into this category
\cite{DKP2015}.

We next consider the parameterized complexity of evaluating queries at
higher levels. The parameterized class \emph{uniform}-\xnl{} is the
class of parameterized problems $Q$ for which there exists a
computable function $f:\Nat \to \Nat$ and a non-deterministic
algorithm that, given a pair $(x,k)$, decides if $(x,k) \in Q$ in
space at most $f(k) \log |x|$ \cite[Proposition 18]{CFG2003}.
\begin{theorem}
    \label{thm:xnlUpperBound}
    Evaluating REWB queries, with size of the query as parameter, is
    in uniform-\xnl{}.
\end{theorem}
\begin{proof}[Proof idea]
  Let $k$ be the size of the query $e_{i}^{1}$ to be evaluated, on a
  data graph with $n$ nodes. Suppose a pair of nodes is connected by a
  data path $w$ in $L(e_{i}^{1})$. Iterations in $e_{i}^{1}$ can only
  occur inside its $F_{i-1}$ sub-expressions. Hence $w$ consists of at
  most $k$ sub-paths, each sub-path $w_{j}$ in the language of some
  sub-expression $f_{i-1}^{j}$. When $f_{i-1}^{j}$ is considered as a
  standard regular expression over its sub-expressions (in $E_{i-1}$),
  there are no bindings. By a standard pigeon hole principle argument,
  we can infer that $w_{j}$ consists of at most $kn$ sub-paths, each
  one in the language of some sub-expression $e_{i-1}^{1}$. This
  argument can be continued to prove that $w$ is of length at most
  $(k^{2}n)^{i}$. The existence of such a path can be guessed and
  verified by a non-deterministic Turing machine in space
  $\mathcal{O}(ik^{2} \log n)$.
\end{proof}

\subsection{Lower Bounds}
We obtain our lower bounds by reducing various versions of the Boolean
formula satisfiability problem to query evaluation. We begin by
describing a schema for reducing the problem of evaluating a Boolean
formula on a given truth assignment to the problem of evaluating a
query on a data graph. The basic ideas for the gadgets we construct
below are from \cite[proofs of Proposition~2, Theorem~5]{LTV2013}.  We
will need to build on these ideas to address finer questions about the
complexity of query evaluation.

Suppose the propositional atoms used in the Boolean formula are among
$\set{\prvar_{1}, \ldots, \prvar_{n}}$. We use $\prvar_{1}, \ldots,
\prvar_{n}$ also as data values.  An edge labeled
$\binom{\prvarl}{\prvar_{j}}$ indicates the propositional atom
$\prvar_{j}$ occurring in a sub-formula.  The data values $\pol$ and
$\nel$ appear on edges labeled with the letter $\ponet$, to indicate
if a propositional atom appears positively or negatively. The symbol $*$
denotes an arbitrary data value different from all others. We will
assume that the Boolean formula is in negation normal form, i.e.,
negation only appears in front of propositional atoms. This
restriction does not result in loss of generality, since any Boolean
formula can be converted into an equi-satisfiable one in negation
normal form with at most linear blowup in the size. The data graph is
a series parallel digraph with a source and a sink, defined as follows
by induction on the structure of the Boolean formula.
\begin{itemize}
  \item Positively occurring propositional atom
    $\prvar_{j}$: $\cdot \xrightarrow{\binom{b}{*}} \cdot
    \xrightarrow{\binom{\ponet}{\pol}} \cdot
    \xrightarrow{\binom{\prvarl}{\prvar_{j}}} \cdot
    \xrightarrow{\binom{e}{*}} \cdot$.
  \item Negatively occurring propositional atom
    $\prvar_{j}$: $\cdot \xrightarrow{\binom{b}{*}} \cdot
    \xrightarrow{\binom{\ponet}{\nel}} \cdot
    \xrightarrow{\binom{\prvarl}{\prvar_{j}}} \cdot
    \xrightarrow{\binom{e}{*}} \cdot$.
  \item $\phi_{1} \land \cdots \land \phi_{r}$: inductively construct
    the data graphs for the conjuncts, then do a standard serial
    composition, by fusing the sink of one graph with the source of the
    next one.
  \item $\phi_{1} \lor \cdots \lor \phi_{r}$: inductively construct
    the data graphs for the disjuncts, then do a standard parallel
    composition, by fusing all the sources into one node and all the
    sinks into another node.
  \item After the whole formula is handled, the source of the
      resulting graph is fused with the sink of the following graph:
      $\cdot \xrightarrow{\binom{a}{\pol}} \cdot
      \xrightarrow{\binom{a}{\nel}} \cdot$.
\end{itemize}
Let $G_{\phi}$ denote the data graph
constructed above for formula $\phi$. The data graph $G_{\phi}$ is
shown below for
$\phi = (\prvar_{1} \lor \neg \prvar_{2}) \land ( (\prvar_{2} \land
\prvar_{3}) \lor (\neg \prvar_{1} \land \prvar_{4}) )$.
    \begin{center}
        \begin{tikzpicture}[>=stealth]
            \node[state] (s1) at (0\ml,0\ml) {};
            \node[state] (s2) at ([xshift=1\ml]s1) {};
            \node[state] (s3) at ([xshift=1\ml]s2) {};
            \node[state] (s4) at ([yshift=0.5\ml,xshift=0.8\ml]s3) {};
            \node[state] (s5) at ([xshift=1\ml]s4) {};
            \node[state] (s6) at ([xshift=1\ml]s5) {};
            \node[state] (s7) at ([yshift=-0.5\ml,xshift=0.8\ml]s6) {};
            \node[state] (s8) at ([yshift=-0.5\ml,xshift=0.8\ml]s3) {};
            \node[state] (s9) at ([xshift=1\ml]s8) {};
            \node[state] (s10) at ([xshift=1\ml]s9) {};
            \node[state] (s11) at ([yshift=0.5\ml,xshift=0.8\ml]s7) {};
            \node[state] (s12) at ([xshift=1\ml]s11) {};
            \node[state] (s13) at ([xshift=1\ml]s12) {};
            \node[state] (s14) at ([xshift=1\ml]s13) {};
            \node[state] (s15) at ([xshift=1\ml]s14) {};
            \node[state] (s16) at ([xshift=1\ml]s15) {};
            \node[state] (s17) at ([xshift=1\ml]s16) {};
            \node[state] (s18) at ([yshift=-0.5\ml,xshift=0.8\ml]s17) {};
            \node[state] (s19) at ([yshift=-0.5\ml,xshift=0.8\ml]s7) {};
            \node[state] (s20) at ([xshift=1\ml]s19) {};
            \node[state] (s21) at ([xshift=1\ml]s20) {};
            \node[state] (s22) at ([xshift=1\ml]s21) {};
            \node[state] (s23) at ([xshift=1\ml]s22) {};
            \node[state] (s24) at ([xshift=1\ml]s23) {};
            \node[state] (s25) at ([xshift=1\ml]s24) {};

            \draw[->] (s1) -- node[auto=left] {$\binom{a}{\pol}$} (s2);
            \draw[->] (s2) -- node[auto=left] {$\binom{a}{\nel}$} (s3);
            \draw[->] (s3) --  (s4);
            \draw[->] (s4) -- node[auto=left] {$\binom{\ponet}{\pol}$} (s5);
            \draw[->] (s5) -- node[auto=left] {$\binom{\prvarl}{\prvar_{1}}$} (s6);
            \draw[->] (s6) --  (s7);
            \draw[->] (s7) -- (s11);
            \draw[->] (s11) -- node[auto=left] {$\binom{\ponet}{\pol}$} (s12);
            \draw[->] (s12) -- node[auto=left] {$\binom{\prvarl}{\prvar_{2}}$} (s13);
            \draw[->] (s13) -- node[auto=left] {$\binom{e}{*}$} (s14);
            \draw[->] (s14) -- node[auto=left] {$\binom{b}{*}$} (s15);
            \draw[->] (s15) -- node[auto=left] {$\binom{\ponet}{\pol}$} (s16);
            \draw[->] (s16) -- node[auto=left] {$\binom{\prvarl}{\prvar_{3}}$} (s17);
            \draw[->] (s17) -- (s18);
            \draw[->] (s3) --  (s8);
            \node at ([xshift=0.5\ml]s3) {$\binom{b}{*}$};
            \draw[->] (s8) -- node[auto=left] {$\binom{\ponet}{\nel}$} (s9);
            \draw[->] (s9) -- node[auto=left] {$\binom{\prvarl}{\prvar_{2}}$} (s10);
            \draw[->] (s10) -- (s7);
            \node at ([xshift=-0.5\ml]s7) {$\binom{e}{*}$};
            \draw[->] (s7) -- (s19);
            \node at ([xshift=0.5\ml]s7) {$\binom{b}{*}$};
            \draw[->] (s19) -- node[auto=left] {$\binom{\ponet}{\nel}$} (s20);
            \draw[->] (s20) -- node[auto=left] {$\binom{\prvarl}{\prvar_{1}}$} (s21);
            \draw[->] (s21) -- node[auto=left] {$\binom{e}{*}$} (s22);
            \draw[->] (s22) -- node[auto=left] {$\binom{b}{*}$} (s23);
            \draw[->] (s23) -- node[auto=left] {$\binom{\ponet}{\pol}$} (s24);
            \draw[->] (s24) -- node[auto=left] {$\binom{\prvarl}{\prvar_{4}}$} (s25);
            \draw[->] (s25) -- (s18);
            \node at ([xshift=-0.5\ml]s18) {$\binom{e}{*}$};
        \end{tikzpicture}
    \end{center}
The query uses $x_{1},
\ldots, x_{k}$ to remember the propositional atoms that are set to
true.
\begin{align}
    e_{\eval}[k] &:= a\bind{x_{\pol}} ( a\bind{x_{\nel}}(\\ \nonumber
    &(b(\ponet[x_{\pol}^{=}]\cdot \prvarl[x_{1}^{=} \lor \cdots \lor
x_{k}^{=}] + \ponet[x_{\nel}^{=}]\cdot \prvarl[x_{1}^{\ne} \land
\cdots \land x_{k}^{\ne}]) e)^{*} \enspace )) \enspace .
\end{align}
\begin{lemma}
  \label{lem:BoolFormulaSchema}
  Let $\phi$ be a Boolean formula over the propositional atoms
  $\prvar_{1}, \ldots, \prvar_{n}$ and $\val: \set{x_{1}, \ldots,
  x_{k}} \to \set{\prvar_{1}, \ldots, \prvar_{n}, *}$ be a valuation.
  The source of $G_{\phi}$ is connected to its sink by a data path
  in $L(e_{\eval}[k], \val)$ iff $\phi$ is satisfied by the truth
  assignment that sets exactly the propositions in $\set{\prvar_{1},
  \ldots, \prvar_{n}} \cap \range(\val)$ to true.
\end{lemma}
\begin{proof}[Proof idea]
  The two bindings in the beginning of $e_{\eval}[k]$ forces
  $x_{\pol}, x_{\nel}$ to contain $\pol, \nel$ respectively. A
  positively occurring propositional atom generates a data path of the
  form $\cdot \xrightarrow{\binom{b}{*}} \cdot
    \xrightarrow{\binom{\ponet}{\pol}} \cdot
    \xrightarrow{\binom{\prvarl}{\prvar_{j}}} \cdot
    \xrightarrow{\binom{e}{*}} \cdot$, which can only be in the
    language of the expression $b\cdot\ponet[x_{\pol}^{=}]\cdot
    \prvarl[x_{1}^{=} \lor \cdots \lor x_{k}^{=}]e$. This forces
    $\prvar_{j}$ to be contained in one of $x_{1}, \ldots,
    x_{k}$. Similar arguments works for negatively occurring atoms.
    Rest of the proof is by induction on the structure of the
    formula.
\end{proof}

\begin{theorem}
  \label{thm:evalLevel1NPHard}
  For queries in $E_{1}$, the evaluation problem is \NP{}-hard.
\end{theorem}
\begin{proof}[Proof idea]
  To check if a Boolean formula $\phi$ is satisfiable, evaluate the
  query $a \bind{x_{1}} a \bind{x_{2}} \cdots a \bind{x_{n}}
  e_{\eval}[n]$ on the data graph $\cdot \xra{\binom{a}{\prvar_{1}/*}}
  \cdot \xra{\binom{a}{\prvar_{2}/*}}\cdots
  \xra{\binom{a}{\prvar_{n}/*}} \cdot - G_{\phi} \rightarrow \cdot$.
  Here, $\xra{\binom{a}{\prvar_{j}/*}}$ denotes two edges in
  parallel, one labeled with $\binom{a}{\prvar_{j}}$ and another with
  $\binom{a}{*}$.
\end{proof}

Evaluating queries in $E_{1}$ is \NP{}-complete, evaluating REWB
queries in general is \PSPACE{}-complete and evaluating queries in
$E_{i}$ is in $\S_{i}$. To prove a corresponding $\S_{i}$ lower bound,
one would need to simulate $\S_{i}$ computations using queries with
bounded depth of nesting of iterated bindings. However, this does not
seem to be possible. We take a closer look at this in the rest of the
paper. Finding the
exact complexity of evaluating queries in $E_{i}$ remains open.

We now extend our satisfiability-to-query evaluation schema to handle
Boolean quantifiers.
Let $\prvars=\set{\prvar_{1}, \ldots, \prvar_{n}}$ be a set of
propositional atoms. To handle existential Boolean quantifiers, we
build a new graph and a query. These gadgets build on earlier ideas to
bring out the difference in the role played by the data graph and the
query while reducing satisfiability to query evaluation. The new
graph $G[\exists k/\prvars]\circ G$, is as follows: $\cdot
\xra{\binom{a_{1}}{\prvar_{1}}} \cdot
\xra{\binom{a_{1}}{\prvar_{2}}}\cdots \xra{\binom{a_{1}}{\prvar_{n}}}
\cdot - G \rightarrow \cdot$. We assume that the letter $a_{1}$ is not
used inside $G$, which is equal to $G_{\phi}$ for some Boolean formula
$\phi$. The new query $e[\exists k]\circ e$ is defined as follows:
\begin{align}
e[\exists k]\circ e &:= a_{1}^{*} a_{1} \bind{x_{1}} a_{1}^{*} a_{1}\bind{x_{2}} a_{1}^{*}
\cdots a_{1}^{*} a_{1}\bind{x_{k}} a_{1}^{*} e \label{eq:expressionForExists}
\end{align}
where $e = e_{\eval}[k]$ for some $k \in \Nat$.

We now give a parameterized lower bound for evaluating $E_{1}$
queries. An instance of the \emph{weighted satisfiability} problem
consists of a Boolean formula (not necessarily in Conjunctive Normal
Form) and a parameter $k \in \Nat$. The goal
is to check if the formula is satisfied by a truth assignment  of
weight $k$. The class \wsat{} is the set of all parameterized problems
that are \fpt{}-reducible to the weighted satisfiability problem
(see \cite[Chapter 25]{DF2013}).
\begin{lemma}
    \label{lem:evalLevel1WSatHard}
    Let $\phi$ be a Boolean formula over the set $\prvars$ of
    propositions and $k \in \Nat$. We can construct in
    polynomial time a data graph $G$ and an REWB $e_{1}^{1}$
    satisfying the following conditions.
    \begin{enumerate}
        \item The source of $G$ is connected to its sink by a data
            path in $L(e_{1}^{1})$ iff $\phi$ has a satisfying
            assignment of weight $k$.
        \item The size of $e_{1}^{1}$ depends only on $k$.
    \end{enumerate}
\end{lemma}
\begin{proof}[Proof idea]
    The required data graph is $G[\exists k/\prvars]\circ G_{\phi}$
    and $e_{1}^{1}$ is $e[\exists k]\circ e_{\eval}[k]$. 
    The data path $\cdot \xra{\binom{a_{1}}{\prvar_{1}}} \cdot
   \xra{\binom{a_{1}}{\prvar_{2}}}\cdots
   \xra{\binom{a_{1}}{\prvar_{n}}} \cdot$ in the graph $G[\exists
   k/\prvars]\circ G_{\phi}$ has to be in the language of
  $a_{1}^{*} a_{1} \bind{x_{1}} a_{1}^{*} a_{1}\bind{x_{2}} a_{1}^{*}
  \cdots a_{1}^{*} a_{1}\bind{x_{k}} a_{1}^{*}$. This induces a
  valuation $\val'$ which maps $\{x_1, \dots, x_k\}$ injectively into
  $\prvars$, denoting the $k$ propositions that are set to true.  With
  this the data path continues from the source of $G_\phi$ 
  to its sink. Rest of the proof follows from
  Lemma~\ref{lem:BoolFormulaSchema}. 
\end{proof}
\begin{theorem}
    \label{thm:evalLevel1WSatHard}
    Evaluating REWB queries in $E_{1}$, parameterized by the size of
    the query is hard for \wsat{} under \fpt{} reductions.
\end{theorem}
\begin{proof}
    The reduction given in Lemma~\ref{lem:evalLevel1WSatHard} is a
    \fpt{} reduction from the weighted satisfiability problem to the
    problem of evaluating $E_{1}$ queries , parameterized by
    the size of the query.
\end{proof}

Finally we extend our gadgets to handle universal Boolean quantifiers.
These gadgets build upon the previous ideas and bring out the role of
nested iterated bindings when satisfiability is reduced to query
evaluation. We would first like to check if the source of some graph $G$ is connected to
its sink by a data path in the language of some REWB $e$, for every
possible injective valuation $\val: \set{x_{1}, \ldots, x_{k}} \to
\prvars$. We will now design some data graphs and expressions to
achieve this. Let $\sk$ be a letter not used in $G$. The data graphs
$G_{0}, \ldots, G_{k}$ are as shown in
Figure~\ref{fig:UniversalGadget}.
\begin{figure}[!h]
    \begin{center}
    \begin{tikzpicture}[>=stealth]
        \begin{scope}
            \node[state] (sop) at (0\ml,0\ml) {};
            \node[state] (sip) at ([yshift=-3\ml]sop) {};
            \node at ([xshift=-0.9\ml]sop) {source};
            \node at ([xshift=-0.7\ml]sip) {sink};
            \node (g) at (barycentric cs:sop=1,sip=1) {$G$};
            \node[rectangle, draw=black, fit=(sop) (g) (sip), inner sep=0.05\ml] {};

            \draw[->] (sop) edge[bend right] node[auto=right] {$\binom{\sk}{\prvar_{1}}$} (sip);
            \draw[->] (sop) .. controls ([xshift=-4\ml]barycentric cs:sop=1,sip=1) .. node[auto=right] {$\binom{\sk}{\prvar_{n}}$} (sip);
            \draw[dotted] ([xshift=-1.5\ml]barycentric cs:sop=1,sip=1) -- ([xshift=-3\ml]barycentric cs:sop=1,sip=1);
        \end{scope}
        \begin{scope}[xshift=1.8\ml, yshift=-1.6\ml, rotate=90]
            \node[state] (a1) at (0\ml,0\ml) {};
            \node[state] (a2) at ([yshift=-1.2\ml]a1) {};
            \node[state] (a3) at ([yshift=-1.7\ml]a2) {};
            \node[state] (sop) at ([xshift=1.2\ml]barycentric cs:a1=1,a3=1) {};
            \node[state] (b1) at ([yshift=-1.5\ml]a3) {};
            \node[state] (b2) at ([yshift=-1\ml]b1) {};
            \node[state] (b3) at ([yshift=-1.7\ml]b2) {};
            \node[state] (sip) at ([xshift=1.2\ml]barycentric cs:b1=1,b3=1) {};
            \node (gi) at (barycentric cs:sop=1,sip=1) {$G_{i-1}$};
            \draw ([xshift=0.3\ml, yshift=0.3\ml]sop.center) rectangle ([xshift=-0.3\ml, yshift=-0.3\ml]sip.center);
            \node[state, label=-90:source] (b) at ([xshift=-1.5\ml]a1) {};
            \node[state, label=-90:sink] (e) at ([xshift=-1.5\ml]b1) {};

            \draw[->, rounded corners] (a1) -- node[auto=left] {$\binom{a_{i}}{\prvar_{1}}$} ([xshift=1\ml]a1.center) -- (sop);
            \draw[->, rounded corners] (a2) -- node[auto=left] {$\binom{a_{i}}{\prvar_{2}}$} ([xshift=1\ml]a2.center) -- (sop);
            \draw[dotted] (a2) -- (a3);
            \draw[->, rounded corners] (a3) -- node[auto=left] {$\binom{a_{i}}{\prvar_{n}}$} ([xshift=1\ml]a3.center) -- (sop);

            \draw[->, rounded corners] (sip) -- ([xshift=1\ml]b1.center) -- node[auto=right] {$\binom{a_{i}}{\prvar_{n}}$} (b1);
            \draw[->, rounded corners] (sip) -- ([xshift=1\ml]b2.center) -- node[auto=left] {$\binom{a_{i}}{\prvar_{n-1}}$} (b2);
            \draw[dotted] (b2) -- (b3);
            \draw[->, rounded corners] (sip) -- ([xshift=1\ml]b3.center) -- node[auto=left] {$\binom{a_{i}}{\prvar_{1}}$} (b3);

            \draw[->] (b) -- node[auto=right] {$b_{i}$} (a1);
            \draw[->] (b1) -- node[auto=left, pos=0.4] {$c_{i}$} (e);
            
            \draw[->] (b2) edge[bend left] node[auto=left, pos=0.6] {$c_{i}$} (a3);
            \draw[->] (b3) edge[bend left] node[auto=left, pos=0.6] {$c_{i}$} (a2);
        \end{scope}
    \end{tikzpicture}
    \end{center}
    \caption{Data graphs $G_{0}$ (left) and $G_{i}$ (right)}
    \label{fig:UniversalGadget}
\end{figure}
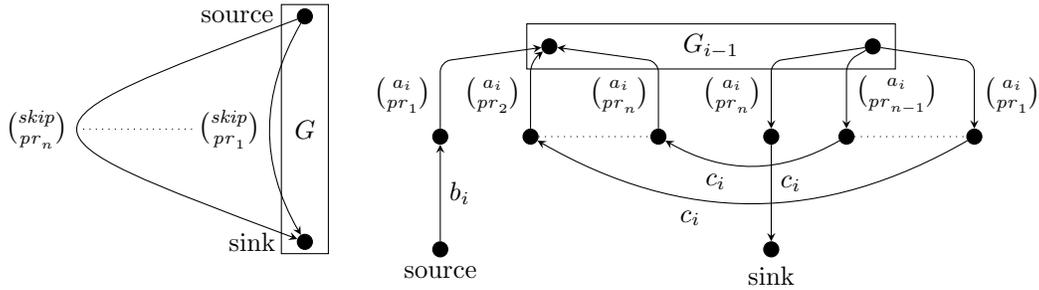
The expressions $e^{0}, \ldots, e^{k}$ are as follows.
\begin{align}
    e^{0} &:= e + \bigplus_{1 \le i < j \le k}\sk[x_{i}^{=} \land
    x_{j}^{=}] &
 e^{i} &:=
b_{i}(a_{i}\bind{x_{i}}(e^{i-1}a_{i}[x_{i}^{=}])c_{i})^{*}
    \label{eq:queryForUniversalVars}
\end{align}
The graph $G_{0}$ and the expression $e^{0}$ are designed to ensure
that the source of $G$ is connected to its sink by a path in
$L(e,\val)$, unless $\val$ is not injective, in which case
$G$ can be bypassed by one of the edges labeled
$\binom{\sk}{\prvar_{j}}$ introduced in $G_{0}$. The graph
$G_{i}$ and the expression $e_{i}$ are designed to ensure that any
path from the source of $G_{i}$ to its sink has to go through
$G_{i-1}$ multiple times, once for each $\prvar_{j}$ stored in the
variable $x_{i}$.
The nesting depth of iterated bindings in the expression $e^{i}$ is
one more than that of $e^{i-1}$.

Suppose $\val$ is a partial valuation of some variables,
whose domain does not intersect with $\set{x_{1}, \ldots, x_{k}}$. We
denote by $\val[\set{x_{1}, \ldots, x_{k}} \to \prvars]$ the set of
valuations $\val'$ that extend $\val$ such that
$\mathrm{domain}(\val') = \mathrm{domain}(\val) \cup \set{x_{1},
\ldots, x_{k}}$ and $\set{\val'(x_{1}), \ldots, \val'(x_{k})}
\subseteq \prvars$. We additionally require that $\val'$ is
injective on $\set{x_{1}, \ldots, x_{k}}$ when we write
$\val[\set{x_{1}, \ldots, x_{k}} \xra{1:1} \prvars]$.
\begin{lemma}
    \label{lem:universalGadgetInduction}
    Let $i \in \set{1,\ldots, k}$ and $\val_{i}$ be a valuation for
    $\fv(e^{i}) \setminus \set{x_{1}, \ldots, x_{i}}$. The source of
    $G_{i}$ is connected to its sink by a data path in $L(e^{i},
    \val_{i})$ iff for every $\val \in \val_{i}[\set{x_{1}, \ldots,
    x_{i}} \to \prvars]$, there is a data path in $L(e^{0},\val)$
    connecting the source of $G_{0}$ to its sink.
\end{lemma}
\begin{proof}[Proof idea]
  The data path has to begin with $b_{i}
  \binom{a_{i}}{\prvar_{1}}$ in the language of
  $b_{i}a_{i}\bind{x_{i}}$, forcing $x_{i}$ to store
  $\prvar_{1}$. Then the path has to traverse $G_{i-1}$ using
  $e^{i-1}$. At the sink of $G_{i-1}$, the path is forced to take
  $\binom{a_{i}}{\prvar_{1}} c_{i}$ to satisfy the condition in
  $a_{i}[x_{i}^{=}] c_{i}$. This forces the path to start again in
  $\binom{a_{i}}{\prvar_{2}}$ and so on.
\end{proof}
We write $G[\forall k/\prvars]\circ G$ and $e[\forall k]\circ e$ to
denote the graph $G_{k}$ and REWB $e^{k}$ constructed above. We
implicitly assume that the variables $x_{1}, \ldots, x_{k}$ are not
bound inside $e$. We can always rename variables to ensure this. If
$e$ is in $E_i$, then $e[\forall k]\circ e$ is in $F_{i+k-1}$.
\begin{lemma}
    \label{lem:universalGadgetCoversAll}
    Let $\val$ be a valuation for $\fv(e) \setminus \set{x_{1},
    \ldots, x_{k}}$ for some REWB $e$. The source of $G[\forall
    k/\prvars]\circ G$ is connected to its sink by a data path in
    $L(e[\forall k]\circ e, \val)$ iff for all
    $\val' \in \val[\set{x_{1}, \ldots, x_{k}} \xra{1:1} \prvars]$,
    the source of $G$ is connected to its sink by a data path in
    $L(e,\val')$.
\end{lemma}
\begin{proof}[Proof idea]
  Lemma~\ref{lem:universalGadgetInduction} ensures that there is a
  path $w_{\val'}$ in $L(e^{0}, \val')$ connecting the source of
  $G_{0}$ to its sink for every valuation $\val' \in \val[\set{x_{1},
  \ldots, x_{k}} \to \prvars]$. From Figure~\ref{fig:UniversalGadget},
  $w_{\val'}$ can either be a $\sk$ edge, or a path through $G$. By
  definition, $e^{0}$ allows a skip edge to be taken only when two
  variables among $x_{1}, \ldots, x_{k}$ have the same data value.
  Hence for valuations $\val'$ that are injective on $\set{x_{1},
  \ldots, x_{k}}$, $w_{\val'}$ is in $L(e,\val')$.
\end{proof}

If $\phi$ is a partially quantified Boolean formula with the
propositional atoms in $\prvars$ occurring freely, we write
$\exists^{k}\prvars~\phi$ to denote that atoms in $\prvars$ are
existentially quantified with the constraint that exactly $k$ of them
should be set to true. We write $\forall^{k}\prvars~\phi$ to denote
that atoms in $\prvars$ are universally quantified and that only those
assignments that set exactly $k$ of the atoms to true are to be
considered. An instance of the \emph{weighted quantified
satisfiability} problem consists of a Boolean formula $\phi$ over the
set $\prvars$ of propositional atoms, a partition $\prvars_{1},
\ldots, \prvars_{\ell}$ of $\prvars$ and numbers $k_{1}, \ldots,
k_{\ell}$.  The goal is to check if $(\exists^{k_{1}}\prvars_{1}
\forall^{k_{2}}\prvars_{2} \cdots \phi)$ is true.
\begin{lemma}
    \label{lem:evalAWSatHard}
    Given an instance of the weighted quantified satisfiability
    problem, We can construct in polynomial time a data graph $G$ and an
    REWB $e_{1 + k_{2} + k_{4} +\cdots}^{1}$ satisfying the following conditions.
    \begin{enumerate}
        \item The source of $G$ is connected to its sink by a data
            path in $L(e_{1 + k_{2} + k_{4} +\cdots}^{1})$ iff the
            given instance of the weighted quantified satisfiability
            problem is a yes instance.
        \item The size of $e_{1 + k_{2} + k_{4} +\cdots}^{1}$ depends
            only on $k_{1}, \ldots, k_{\ell}$.
    \end{enumerate}
\end{lemma}
\begin{proof}[Proof idea]
    The required data graph $G$ is $G[\exists k_{1}/\prvars_{1}] \circ
    G[\forall k_{2}/\prvars_{2}] \circ \cdots \circ G_{\phi}$ and the
    required REWB $e_{1 + k_{2} + k_{4} +\cdots}^{1}$ is $e[\exists
    k_{1}] \circ e[\forall k_{2}] \circ \cdots \circ e_{\eval}[k_{1} +
    \cdots + k_{\ell}]$. We assume that $\circ$ associates to the
    right, so $G_{1} \circ G_{2} \circ G_{3}$ is $G_{1} \circ (G_{2}
    \circ G_{3})$ and $e^{1} \circ e^{2} \circ e^{3}$ is $e^{1}
    \circ(e^{2} \circ e^{3})$. Correctness follows from
    Lemma~\ref{lem:universalGadgetCoversAll} and
    Lemma~\ref{lem:BoolFormulaSchema}.
\end{proof}

The weighted quantified satisfiability problem is parameterized by
$\ell + k_{1} + \cdots + k_{\ell}$. The class \awsat{} is the set of
parameterized problems that are \fpt{}-reducible to the weighted
quantified satisfiability problem (see \cite[Chapter 26]{DF2013}).
\begin{theorem}
    \label{thm:evalAWSatHard}
    Evaluating REWB queries, parameterized by the size of
    the query is hard for \awsat{} under \fpt{} reductions.
\end{theorem}
\begin{proof}
    The reduction given in Lemma~\ref{lem:evalAWSatHard} is a \fpt{}
reduction from the weighted quantified satisfiability problem to the
problem of evaluating REWB queries, with query size as the parameter.
\end{proof}


\section{Summary and Open Problems}
We have proved that increasing the depth of nesting of iterated
bindings in REWBs increase expressiveness. Given an REWB, it is
undecidable to check if its language can be defined with another REWB
with smaller depth of nesting of iterated bindings. The complexity of
query evaluation problems are summarized in the following table,
followed by a list of technical challenges to be overcome for closing
the gaps.
\begin{center}
\begin{tabular}[h]{lll}
  \toprule
  Query level & Evaluation & Parameterized complexity,
  query size is parameter\\
  \midrule
  $E_{1}$ & \NP{}-complete & (?2)\wsat{} lower bound, \wop{} upper bound\\
  $E_{i}$, $i > 1$ & (?1), $\S_{i}$ upper bound & (?3)\\
  Unbounded & \PSPACE{}-complete \cite{LTV2013} & (?4)\awsat{} lower
  bound, uniform-\xnl{} upper bound\\
  \bottomrule
\end{tabular}
\end{center}
\begin{enumerate}
    \item Suppose we want to check the satisfiability of a
        $\S_{2}$ Boolean
    formula over $(n_{e} + n_{u})$ propositional atoms of which the
    first $n_{e}$ atoms are existentially quantified and the last
    $n_{u}$ are universally quantified. With currently known
    techniques, reducing this to query evaluation results in an REWB
    in $E_{(n_{u}+1)}$. Hence, with bounded nesting depth, we cannot
    even prove a $\S_{2}$ lower bound.
  \item Weighted formula satisfiability, complete for \wsat{}, can be
    simulated with series-parallel graphs. Queries in $E_{1}$ do not
    seem to be powerful enough for weighted circuits.
  \item Without parameterization, the $\S_{i}$ upper bound is obtained
    by an oracle hierarchy of \NP{} machines. With parameterization,
    an oracle hierarchy of \wop{} machines does not correspond to any
    parameterized complexity class. See \cite[Section 4]{CFG2003} for
    discussions on subtle points which make classical complexity
    results fail in parameterized complexity.
  \item As in point 2, here one might hope for a \awp{} lower bound,
    which is quantified weighted circuit satisfiability (stronger than
    \awsat{}, which is quantified weighted formula satisfiability).
    Even if this improvement can be made, there is another classical
    complexity result not having analogous result in parameterized
    complexity: not much is known about the relation between
    parameterized alternating time bounded class (\awp{}) and
    parameterized space bounded class (uniform-\xnl{}).
\end{enumerate}

\paragraph*{Acknowledgements}
The authors thank Partha Mukhopadhyay and Geevarghese Philip for
helpful discussions about polynomial time hierarchy and parameterized
complexity theory.


\bibliographystyle{plain}
\bibliography{references}

\newpage
\appendix
\section{Details of Section~\ref{sec:nest-depth-iter-short}}
\label{sec:append-sect-hierarchy}

This section contains complete proofs and explanations from
Section~\ref{sec:nest-depth-iter-short}. We start with the semantics
of the automaton view of expressions. 

\subsection{Automata View of Expressions}

We will now provide in more detail the semantics of the automata
defined for each expression. Expressions we consider contain some free
variables and some bound variables due to the $\bind{x}$
operator. Without loss of generality, we will assume that no two
occurrences of the binding operator contain the same variable
name. Recall that for an expression $e$, we denote the set of its free
variables by $\fv(e)$, and the set of all variables (free and bound)
by $\var(e)$. A valuation associates every free variable to a data
value.

Consider an expression $f_i^1$, and its corresponding automaton
$\aut{f_i^1}$. Let $\val$ be a valuation which associates a data value
to all the free variables $\fv(f_i^1)$ of $f_i^1$. A run of
$\aut{f_i^1}$ over a data word $w =
\binom{a_1}{d_1}\binom{a_2}{d_2}\dots \binom{a_n}{d_n}$ given
valuation $\val$ is as follows:
\begin{align*}
  q_0 \xra{~w_1~} q_1 \xra{~w_2~} ~\cdots~ \xra{~w_{m-1}~} q_{m-1}
  \xra{~w_m~}q_m
\end{align*}
where:
\begin{itemize}
\item $q_0$ is an initial state,
\item $w = w_1 w_2 \dots w_m$,
\item if $i = 0$ then $m = n$ and for each $j$, we have $w_j =
  \binom{a_j}{d_j}$. Moreover, for each $j$, there exists a transition
  $q_{j-1} \xra{~a_j~} q_j$ or $q_{j-1} \xra{~a_j[c]~} q_j$ such that
  $d_j, \val \models c$,
\item if $i > 0$, then for each $j$ there exists a transition $q_{j-1}
  \xra{~e_{i}^{j}~} q_j$ such that $w_j \in L(e_{i}^j,
  \val\upharpoonright e_i^j)$, where $\val \upharpoonright e_i^j$
  denotes the valuation restricted to $\fv(e_i^j)$.
\end{itemize}

The run is accepting if $q_m$ is an accepting state of the automaton.
The language $L(\aut{f_i^1}, \val)$ is the set of words for which
$\aut{f_i^1}$ has an accepting run given valuation
$\val$.

Given an expression $e_i^1$ and a valuation $\val$ of its free
variables, the run of $\aut{e_i^1}$ on a data word $w$ is defined as:
\begin{align*}
  (q_0, \val_0) \xra{~w_1~} (q_1, \val_1) \xra{~w_2~} \cdots
  \xra{~w_{m-1}~} (q_{m-1}, \val_{m-1}) \xra{~w_m~} (q_m, \val_m)
\end{align*}
where
\begin{itemize}
\item $q_0$ is an initial state,
\item each $w_j$ is a data word such that $w=w_1w_2\dots w_m$,
\item each $\val_j$ is a partial function from $\var(e_i^1)$ to the
  set of data values, with $\val_0 = \val$;
\item for each $j$, either $w_j = \binom{a}{d}$ and there is a
  transition $q_{j-1} \xra{~a\bind{x}~} q_{j}$ and $\val_{j} =
  \val_{j-1}[x \rightarrow d]$, or there is a transition $q_{j-1}
  \xra{~f_{i-1}^{j}~} q_{j}$ with $w_{j} \in L(f_{i-1}^{j}, \val_{j-1}
  \upharpoonright f_{i-1}^j)$ and $\val_{j} = \val_{j-1}$. As before,
  $\val_{j-1} \upharpoonright f_{i-1}^j$ is a valuation for
  $f_{i-1}^j$ obtained by restricting the partial function
  $\val_{j-1}$ to $\fv(f_{i-1}^j)$.
\end{itemize}

The notion of acceptance and language $L(\aut{e_i^1}, \val)$ are
defined in a way similar to the $F_i$ case. We now explain with an
example the necessity of the restriction that no two occurrences of
the binding operator contain the same variable name. Suppose
$e_{1}^{1} = a\bind{x}(b\bind{x}(c[x^{=}])\cdot c[x^{\ne}])$. An
automaton would have the transitions $q_{0} \xra{~a\bind{x}~} q_{1}
\xra{~b\bind{x}~} q_{2} \xra{~c[x^{=}]} q_{3} \xra{~c[x^{\ne}]}
q_{4}$. There is no elegant way to specify that the value to be tested
in the transition $q_{3} \xra{~c[x^{\ne}]} q_{4}$ is the one bound in
$q_{0} \xra{~a\bind{x}~} q_{1}$ and not the one bound in $q_{1}
\xra{~b\bind{x}~} q_{2}$. Hence we consider the language equivalent
expression $a\bind{x_{1}}(b\bind{x_{2}}(c[x_{2}^{=}])\cdot
c[x_{1}^{\ne}])$, which avoids this problem.

The following lemma can be shown by an induction on $i$.

\newtheorem*{lemmaA1}{Lemma A.1}
\begin{lemmaA1}
  \label{lem:automaton-expression-same}
  For every expression $f_i^1$, and for every valuation $\val$ of
  $\fv(f_i^1)$, the languages $L(f_i^1, \val)$ and $L(\aut{f_i^1},
  \val)$ are equal. Similarly for expressions $e_i^1$.
\end{lemmaA1}

For any expression $e$, the size of $\aut{e}$ is defined as the number
of states present.

\subsection{Strictness of the Hierarchy}

Using the semantics of the automata developed above, we will give a
full proof of Lemma~\ref{lem:u-words-property}.

\smallskip
\Repeat{Lemma}{lem:u-words-property}
  Let $e_{i}^1$ be an expression and let $n \in \Nat$ be greater than
  $(|\aut{e}|+1)$ and $(|\var(e)|+1)$ for every sub-expression $e$ of
  $e_i^1$. Let $\val$ be a valuation of $\fv(e_{i}^1)$ and let $x, z$
  be data words. Then: $xu_{i,n} z \in L(e_{i}^1, \val) ~~\implies~~
  x\bar{u}_{i,n} z \in L(e_{i}^1, \val) \text{ for some }
  \bar{u}_{i,n} \in \mismatch_{i,n}$.

\begin{proof}
  We proceed by an induction on $i$. We start with the base
  case. Suppose $xu_{1,n}z \in L(e_1^1, \val)$ for some expression
  $e_1^1$ with $\max(|\aut{e_1^1}|,|\var(e_1^1)|) < n$. The automaton
  $\aut{e_1^1}$ has an accepting run of the following form:
  \begin{align*}
    \text{Run } \rho_1:~~ (q_0, \val_0) \xra{~w_1~} (q_1, \val_1)
    \xra{~w_2~} \cdots \xra{~w_{m-1}~} (q_{m-1},\val_{m-1})
    \xra{~w_m~} (q_m, \val_m)
  \end{align*}
  where $xu_{1, n}z = w_1w_2\dots w_m$. Recall that automata for
  $E_1$-expressions are acyclic, so states cannot repeat in a run.
  Since the number of states is strictly less than $n$ and $u_{1,n}$
  contains $n^2$ occurrences of $a_1b_1$, there is some $w_p$ which
  contains $n$ occurrences of $a_1 b_1$:
  \begin{align*}
    \binom{a_1}{d[1,{j+1}]}\binom{b_1}{d[1,{j+1}]} \dots
    \binom{a_1}{d[1,{j+n}]}\binom{b_1}{d[1,{j+n}]}
  \end{align*}
  Then by definition of runs of $\aut{e_1^1}$:
  \begin{itemize}
  \item there is a transition $q_{p-1} \xra{f_0^1} q_p$ in
    $\aut{e_1^1}$ with $w_p \in L(f_0^1, \val_{p-1})$ and
  \item valuation $\val_p$ equals $\val_{p-1}$ since the transition
    $q_{p-1} \xra{~~} q_p$ does not contain a binding.
  \end{itemize}

  
  Note that $\range(\nu_{p-1})$ can contain at most $|\var(e_1^1)|$
  distinct data values. Since by assumption $n$ is strictly bigger
  $|\var(e_1^1)|$, there can be at most $n-1$ distinct data values in
  $\range(\val_{p-1})$.   

  Let us now zoom in to the accepting run of
  $\aut{f_0^1}$ on the sub-word $w_p$.
  \begin{align*}
    \text{Run } \sigma_1:~~ q_0' \xra{~w_1'~} q_1' \xra{~w_2'~}
    ~\cdots~ \xra{~w_{s-1}'~} q'_{s-1} \xra{~w'_s~} q'_s
  \end{align*}
  with $w_p = w_1'w_2'\dots w_s'$. Each transition reads a single
  letter: that is, it is of the form $q'_{j-1} \xra{~a~} q'_j$ or
  $q_{j-1} \xra{~a[c]~} q_j$ with letter $a$ denoting either $a_1$ or
  $b_1$. Note that there can be no further bindings in this level
  $F_0$. So, each condition $a[c]$ in a transition can check for
  equality or inequality with respect to data values in
  $\range(\val_{p-1})$. Consider the transitions reading $b_1$. As
  there are at most $n-1$ data values in $\range(\val_{p-1})$, there
  is some $\binom{b_1}{d[1,j']}$ in $w_p$ such that $d[1,j']$ is
  different from all values in $\range(\val_{p-1})$. Therefore,
  changing $d[1,j']$ in $\binom{b_1}{d[1,j']}$ to a new data value
  $d'[1,j'] \notin \range(\val_{p-1})$ will give a data word which
  continues to satisfy all conditions occurring in the run $\sigma_1$
  of $\aut{f_0^1}$. The run $\rho_1$ is oblivious to this change. This
  is because there are no bindings in $\sigma_1$ and hence the
  valuation $\val_p$ is the same as $\val_{p-1}$. Hence the same run
  $\rho_1$ of $\aut{e_1^1}$ accepts this modified word.  Observe that
  in this word there is a mismatch between an $a_1$ and the
  consecutive $b_1$ occurring in $u_{1,n}$ and is of the form $x
  \bar{u}_{1,n}z$ as required by the lemma. This proves the lemma for
  the base case $i=1$.

  We will now prove the induction step. Assume that the lemma is true
  for some $i-1$. We will now prove it for $i$. Consider the word
  $xu_{i, n}z$. Suppose it belongs to $L(e_{i}^1, \val)$ for some
  expression $e_{i}^1$ with the value $n$ being an upper bound on
  $|\aut{e}|+1$ and $|\var(e)|+1$ for every subexpression $e$ of
  $e_i^1$. Let $\rho_i$ be the accepting run of $\aut{e_{i}^1}$ on
  $xu_{i,n}z$:
  \begin{align*}
    \text{Run } \rho_i: ~~ (q_0, \val_0) \xra{~w_1~} (q_1, \val_1)
    \xra{~w_2~} \cdots \xra{~w_{m-1}~} (q_{m-1},\val_{m-1})
    \xra{~w_m~} (q_m, \val_m)
  \end{align*}
  with $xu_{i,n}z = w_1w_2\dots w_m$. Since the automaton
  $\aut{e_i^1}$ is acyclic, no state can repeat in $\rho_i$.  As the
  number of states is less than $n$ and $u_{i,n}$ contains $n^2$
  occurrences of $a_i u_{i-1,n} b_i$, some $w_p$ contains $n$
  occurrences of the block $a_iu_{i-1,n} b_i$:
  \begin{align*}
    \binom{a_i}{d[i,j+1]} u_{i-1, n} \binom{b_i}{d[i,j+1]} ~\cdots~
    \binom{a_i}{d[i,j+n]} u_{i-1,n} \binom{b_i}{d[i,j+n]}
  \end{align*}
  Then, by definition of runs of $\aut{e_i^1}$:
  \begin{itemize}
  \item there is a transition $q_{p-1} \xra{f_{i-1}^1} q_p$ in
    $\aut{e_{i}^1}$ with $w_p \in L(f_{i-1}^1, \val_{p-1})$ and
  \item $\val_p = \val_{p-1}$.
  \end{itemize}

  As $\range(\val_{p-1})$ can contain at most $|\var(e_i^1)|$ distinct
  data values, and since $n$ is bigger than $|var(e_i^1)| + 1$, we
  observe that
  $\range(\val_{p-1})$ contains at most $n-1$ distinct data values.
  Consider the run of $\aut{f_{i-1}^1}$ on $w_p$ given valuation
  $\val_{p-1}$:
  \begin{align*}
    \text{Run } \s_i: ~~ q'_0 \xra{~w'_1~} q'_1 \xra{~w'_2~} ~\cdots~ \xra{~w'_{s-1}~}
    q'_{s-1} \xra{~w'_s~} q'_s
  \end{align*}
  where $w_p = w'_1w'_2\dots w'_s$.

  If some $w'_j$ contains $u_{i-1, n}$ entirely, then this $w'_j$
  belongs to the language of $L(e_{i-1}^1, \val_{p-1})$ for some
  expression $e_{i-1}^1$. This is as per the definition of runs of
  $F_i$ automata. Additionally, $e_{i-1}^1$ is a subexpression of
  $e_i^1$ and hence satisfies the condition that $n$ is bigger than
  $|\aut{e}|+1$ and $|\var(e)|+1$ for every subexpression $e$ of
  $e_{i-1}^1$. We can then use the induction hypothesis to infer that
  there is a mismatched word in $L(e_{i-1}^1, \val_{p-1})$. Hence we
  can replace $w_j'$ with this mismatched word to obtain the same runs
  $\sigma_i$ and $\rho_i$, thus proving the lemma for this case.

  Otherwise, no $w_j'$ contains both $\binom{a_i}{d[i,j]}$ and
  $\binom{b_i}{d[i,j]}$ of a block $a_iu_{i-1,n}b_i$. As
  $\range(\val_{p-1})$ has at most $n-1$ distinct data values, there
  is one $\binom{b_i}{d[i,j]}$ such that $d[i,j]$ is not present in
  $\range(\val_{p-1})$. Let this be present in $w_k'$, and let
  $e_{i-1}^k$ be the sub-expression in the transition $q'_{k-1}
  \xra{~~} q'_k$ with $w_k' \in L(e_{i-1}^k, \val_{p-1})$. Consider a
  fresh data value $d' \notin \range(\val_{p-1})$ and which is
  different from every data value in $w_k'$. Change all occurrences of
  $d[i,j]$ in $w_k'$ to this fresh data value $d'$. The modified word
  belongs to $L(e_{i-1}^k, \val_{p-1})$ (this can be shown by a
  structural induction for general REWBs). Therefore the run
  $\sigma_i$ holds for the word $w_p$ with $w_k'$ modified to this new
  word. Moreover, as discussed in the base case, the run $\rho_i$ is
  not affected by this change as the valuation $\val_p$ is the same as
  $\val_{p-1}$. This gives a word with a mismatch between an $a_i$ and
  the corresponding $b_i$ that is accepted by $e_i^1$, thereby proving
  the lemma.
 
\end{proof}

\subsection{Undecidability of Membership at a Given Level}

This section is devoted to proof of the following theorem.

\smallskip
\Repeat{Theorem}{thm:undec-height}
  Given an expression in $F_{i+1}$, checking if there exists an
  equivalent expression in $F_i$ is undecidable.
\smallskip

The basic idea is
from the proof of undecidability of universality of REWBs and related
formalisms \cite{NSV2004,LTV2013}. If a given REWB is universal, i.e.,
accepts all data words, then there is a language equivalent expression
that does not use any binding. The undecidability of universality can
hence be interpreted to mean that determining the usefulness of
bindings in an expression is undecidable. We combine this insight with
results we have obtained for the expressions $r_{i}$ in the previous
sub-section to prove Theorem~\ref{thm:undec-height}. We proceed by a
reduction from Post's Correspondence Problem (PCP). An instance of PCP
is a set $\{(u_1, v_1), (u_2, v_2), \dots, (u_n, v_n)\}$ of pairs of
words over a finite alphabet $\Sigma_{PCP}$. A solution to this
instance is a sequence $l_1, l_2, \dots, l_m$ with each $l_j \in \{1,
\dots, n\}$ such that $u_{l_1}u_{l_2}\dots u_{l_m} =
v_{l_1}v_{l_2}\dots v_{l_m}$.

Suppose we are given an instance $\{(u_1, v_1), (u_2, v_2), \dots,
(u_n, v_n)\}$ of PCP. We will encode a solution $l_1, \dots, l_m$ to
this instance by a set of data words of the form:
\begin{align*}
  \th_1 ~\binom{\#}{d_1} ~z~ \binom{\#}{d_2}~ \th_2
\end{align*} where:
\begin{itemize}
\item $z \in L(r_i)$, with $r_i$ being the expression in
  Definition~\ref{def:strict-expressions},
\item $\theta_1$ is the data word:
  \begin{align*}
    \binom{\$_{l_1}}{h_1}\binom{\a_1}{1}\dots
    \binom{\a_{p}}{p}\binom{\$_{l_2}}{h_2}\binom{\a_{p+1}}{p+1}~~\dots~~
    \binom{\$_{l_m}}{h_m}~\dots \binom{\a_{r}}{r}
  \end{align*}
  where the word $\a_s \dots \a_{t}$ between $\$_{l_j}$ and
  $\$_{l_{j+1}}$ equals the word $u_{l_j}$,
\item $\theta_2$ is the data word:
  \begin{align*}
    \binom{\$_{l_1}}{h_1}\binom{\b_1}{1}\dots
    \binom{\b_{q}}{q}\binom{\$_{l_2}}{h_2}\binom{\b_{q+1}}{q+1}~~\dots~~
    \binom{\$_{l_m}}{h_m}~\dots \binom{\b_{r}}{r}
  \end{align*}
  where the word $\b_s \dots \b_{t}$ between $\$_{l_j}$ and
  $\$_{l_{j+1}}$ equals the word $v_{l_j}$,
\item the data values $\{1, \dots, r, d_1, d_2, h_1, \dots, h_m\}$ are
  all distinct.
\end{itemize}

We will first construct an expression $\Delta$ that accepts all words
of the form $w_1 \binom{\#}{d_{1}} z
\binom{\#}{d_{2}} w_2$ with $z \in L(r_i)$ such that the part $w_1\#
~~\#w_2$ does not satisfy the conditions mentioned above. The
expression $\Delta$ will be in $E_{i+1}$. We will then reason that
this expression will have an equivalent expression in $F_i$ iff PCP
has no solution.

Let $\G$ denote the finite alphabet $\Sigma_{PCP}\cup\{\$_1, \dots,
\$_n, \#\}$.  We will now exhaustively reason about the situations
when a word $w_1 \binom{\#}{d_{1}} z \binom{\#}{d_{2}} w_2$ is not an
encoding of the PCP solution. This will give us the expression
$\Delta$ mentioned above.

\textbullet\hphantom{a} Projection of the word on to the finite
alphabet is not of the form $(\$_1u_1 + \dots +
\$_nu_n)^*~\#~z~\#~(\$_1v_1 + \dots + \$_nv_n)^*$. Let $\phi_1$ and
$\phi_2$ be the regular expressions denoting the complement of
$(\$_1u_1 + \dots + \$_nu_n)^*$ and $(\$_1v_1 + \dots + \$_nv_n)^*$
respectively. The required expression that accepts words with a
mistake in the finite alphabet is $\phi_1 ~\#~r_i~\#~ \G^* ~+
~\G^*~\#~r_i~\#~\phi_2 $. Note that this expression is at the same
level as $r_i$ since $r_{i}$ is not in the scope of any binding.

\textbullet\hphantom{a} Words where the data values are not according
to the encoding. Firstly, words of the form $ \cdots
\binom{\#}{d}~z~\binom{\#}{d'} \cdots$ where $d$ or $d'$
repeat. Expression accepting words where $d$ repeats is given by:
\begin{align*}
  \bigplus_{a \in \G} (\G^* a\bind{x}(~\G^* ~\#[x^=]~) ~r_i~ \G^* ~+~
  \G^*\# \bind{x}(~r_i~ \G^* a[x^=] ~\G^*))
\end{align*}
A similar expression can be given for the case where $d'$
repeats. Since these expressions add a binding over $r_i$, they are in
$E_{i+1}$. 

\textbullet\hphantom{a} Words of the form $\cdots \binom{*}{d} \cdots
\binom{*}{d} \cdots ~\#~z~\# \cdots$ where a data value repeats before
the $\#~z~\#$ and words of the form $\cdots ~\#~z~\# \cdots
\binom{*}{d} \cdots \binom{*}{d} \cdots$ where values repeat after
$\#~z~\#$
\begin{align*}
  \bigplus_{a,b\in \G} (\G^* a\bind{x}(~\G^*
  ~b[x^=]~\G^*)~\#~r_i~\#~\G^* ~+~ \G^* ~\#~ r_i~\#~ \G^*
  a\bind{x}(~\G^* ~b[x^=]~\G^*))
\end{align*}

\textbullet\hphantom{a} Note that in the encoding of the solution, the
data values in the $j^{th}$ dollar symbol before and after $\#~z~\#$
need to be the same. We will now consider words where this is not
true. Let us first look at words where the mismatch occurs either in
the first dollar symbol or in the last dollar symbol.
\begin{align*}
  \bigplus_{\$, \$' \in \set{\$_{1}, \ldots, \$_{n}}} (\$ \bind{x} (
  \G^* ~\#~r_i~\#~\$' [x^\neq])~\G^* ~~+~~ \G^* \$
  \bind{x}(~\S_{PCP}^*~\#~r_i~\#~ \G^* ~\$' [x^\neq]~) \S_{PCP}^*)
\end{align*}
Suppose the first dollar symbols to the left and right of $\# z \#$
have the same data value, and so do the last dollar symbols. In this
case, if there some $j$ such that the $j$\textsuperscript{th} dollar
symbol to the left and right of $\# z \#$ have different data values,
the data word is of the following form.
\begin{align*}
  \cdots \binom{\d_{1}}{d}~\S_{PCP}^*~ \binom{\d_{2}}{d'} \cdots
  ~\#~z~\#~ \cdots
  \binom{\d_{3}}{d}~\S_{PCP}^*~\binom{\d_{4}}{d''}~\cdots
\end{align*}
where $\d_{1}, \d_{2}, \d_{3}, \d_{4} \in \set{\$_{1}, \ldots,
  \$_{n}}$ and $d' \ne d''$. There is a data value $d$ occurring with
a dollar on both sides, and the data values attached with next dollar
symbols on the two sides do not match. The expression for such words
is given by:
\begin{align*}
  \bigplus_{\d_{1}, \ldots, \d_{4} \in \set{\$_{1}, \ldots, \$_{n} }}
  (\G^* ~\d_{1} \bind{x}~( ~\S_{PCP}^* ~\d_{2} \bind{y} ~(~ \G^*
  ~\#~r_i~\#~ \G^*~ \d_{3} [x^=] ~\S_{PCP}^* ~\d_{4} [y^\neq]~)~ )~
  \G^*)
\end{align*}

\textbullet\hphantom{a} Now we will consider words where the mismatch
of data values occurs in a non-dollar position. We start with the
expression for words with a mismatch in the first or last non-dollar
position:
\begin{align*}
  \bigplus_{\d_{1}, \d_{2} \in \set{\$_{1}, \ldots, \$_{n}}, a, b \in
    \S_{PCP}}\d_{1}~a \bind{x} (\G^* ~\#~r_i~\#~ \d_{2} ~b
  [x^\neq])\G^* ~~ + ~~ \G^* a\bind{x} (\#~r_i~\# ~\G^* ~b[x^\neq])
\end{align*}
For detecting mismatch at an intermediate position, we resort to the
same idea as in the previous case. We consider words of the form:
\begin{align*}
  \cdots \binom{\a_1}{d}~(\e + \d_{1})~\binom{\a_2}{d'}
  \cdots~\#~r_i~\#~ \cdots \binom{\a_3}{d}~(\e +
  \d_{2})~\binom{\a_4}{d''}~\cdots
\end{align*}
The expression for such words is given by:
\begin{align*}
  \bigplus_{\d_{1},\d_{2} \in \set{\$_{1}, \ldots, \$_{n}},
    \a_{1},\ldots, \a_{4} \in \S_{PCP}} \G^* \a_1 \bind{x} ((\e
    + \d_{1})~\a_2 \bind{y} (\G^*\#r_i\#\G^*\a_3 [x^=]~(\e +
    \d_{2})~\a_4
  [y^\neq]))\G^*
\end{align*}

\textbullet\hphantom{a} We are now left with words where the data
values on every corresponding position before and after $~\# ~z~ \#$
match. Among these words, the non-solutions are the ones where for a
particular data value occurring on both sides of $\#~z~\#$, the
corresponding letters do not match.  The expression for such words is
given by:
\begin{align*}
  \bigplus_{\g_{1} \ne \g_{2}}\G^*~\g_1 \bind{x} (~\G^*~\#~r_i~\#~\G^*
  ~\g_2 [x^=]~)~\G^*
\end{align*}

The required expression $\Delta$ is the sum of all the above
expressions. Note that $\D$ has a binding made on the left side of
$\#~r_i~\#$ that is checked on the right side. This makes expression
$\Delta$ to fall in $E_{i+1}$ as expression $r_i$ is in $F_{i}$.

\newtheorem*{lemmaA2}{Lemma A.2}
\begin{lemmaA2}
  The expression $\Delta$ has an equivalent expression in $F_i$ iff
  the given PCP instance has no solution.
\end{lemmaA2}
\begin{proof}
  Suppose PCP instance has no solution. Then all words of the form
  $w_1~\#~z~\#~w_2$ with $w_1, w_2 \in \G^*$ and $z \in L(r_i)$ are in
  the language of the expression $\Delta$. Therefore an equivalent
  expression for $\Delta$ is $\G^*~\#~r_i~\#~\G^*$. This expression is
  in $F_i$ as the expression $r_i$ is in $F_i$.

  Suppose PCP instance has a solution. Let us assume that $\Delta$ has
  an equivalent expression $f_{i}^1$. We will show that this leads to
  a contradiction. For technical convenience, let us assume that
  $f_i^1$ has no free variables (the case with free variables can be
  handled in a similar way). Let $n$ be a natural number such that
  $|\aut{f_i^1}| < n$. Consider a word $\th_1~\#~u_{i,n}~\#~\th_2$
  that encodes the solution of the PCP instance. Let $\th_2'$ be a new
  data word obtained from $\th_2$ by modifying the last data value to
  a fresh data value not occurring in
  $\th_1~\#~u_{i,n}~\#~\th_2$. Then, $\th_1~\#~u_{i,n}~\#~\th'_2$ does
  not encode any solution and hence belongs to $L(\Delta)$.  Let us
  now look at the run of $\aut{f_i^1}$ on the word
  $\th_1~\#~u_{i,n}~\#~\th'_2$:
  \begin{align*}
    q_0 \xra{~w_1~} q_1 \xra{~w_2~} ~\cdots~ \xra{~w_{m-1}~} q_{m-1}
    \xra{~w_m~}q_m
  \end{align*}
  where $w_1w_2 \dots w_m = \th_1~\#~u_{i,n}~\#~\th'_2$.
  
  If some $w_p$ contains $u_{i,n}$ entirely, then by definition of
  runs of $\aut{f_i^1}$, there is some subexpression $e_{i}^1$ and a
  word $x u_{i,n} z$ such that $x u_{i,n} z \in L(e_i^1)$. Note that
  we assumed that there are no free variables. Then, by Lemma
  \ref{lem:u-words-property} there is a word $x \bar{u}_{i,n} z \in
  L(e_i^1)$ where $\bar{u}_{i,n}$ contains a mismatch. However, by
  definition of $\Delta$, this is not possible. Therefore no $w_p$ can
  contain $u_{i,n}$ entirely. This would then imply that $\th_1$ and
  $\th_2'$ lie in different $w_j$: in particular, the last part of the
  run $w_m$ contains the last letter in $\th'_2$ and moreover does not
  contain any part of $\th_1$. Hence, data values in $\th_{2}'$ are
  never compared with those in $\th_{1}$. Note that by the definition
  of runs, the word $w_m \in L(e_i^2)$ for some subexpression $e_i^2$.
  Changing the last data value of $w_{m}$ back to the value in $\th_2$
  will result in a word which is an automorphic copy of $w_m$ and
  hence this modified word should also lie in $L(e_i^2)$. This shows
  that the same run of $\aut{f_{i}^{1}}$ can accept
  $\th_1~\#~u_{i,n}~\#~\th_2$ which encodes a solution of the PCP
  instance. Therefore the expression supposed to be equivalent to
  $\Delta$ accepts a solution of the PCP instance. A contradiction.
\end{proof}

The above lemma proves Theorem~\ref{thm:undec-height}. The expression
$\Delta$ is in $E_{i+1}$ (and hence in $F_{i+1}$). Checking if it has
an equivalent expression in $F_i$ is undecidable as this can encode
PCP.


\section{Details of Section~\ref{sec:query-evaluation}}

\subsection{Upper Bounds}
We first introduce some normal forms for expressions in $E_{i}$. Let
$U_{i}$ be the set of REWBs generated by the grammar $U_{i} ::=
F_{i-1} ~|~ U_{i} \cdot U_{i} ~|~ a\downarrow_{x}(U_{i})$. An
expression in $E_{i}$ is said to be in \emph{Union Normal Form} (UNF)
if it is of the form $u_{i}^{1} + u_{i}^{2} + \cdots + u_{i}^{r}$,
where $u_{i}^{j}$ is an expression in $U_{i}$ for every $j \in
\set{1,\ldots, r}$. From the semantics of REWBs, we infer that binding
and concatenation distribute over union. By repeatedly applying this
fact to any expression in $E_{i}$, we get the following result.
\newtheorem*{proposition1}{Proposition B.1}
\begin{proposition1}
    \label{prop:UnionNormalForm}
    For every expression $e_{i}^{1}$ in $E_{i}$, there exists a language
    equivalent one $u_{i}^{1} + u_{i}^{2} + \cdots + u_{i}^{r}$ in
    UNF such that $|\aut{u_{i}^{j}}| \le
    |\aut{e_{i}^{1}}|$ for every $j \in \set{1,\ldots, r}$.
\end{proposition1}

\Repeat{Lemma}{lem:evalFiWithOracleForEi}
  With an oracle for evaluating $E_{i}$ queries,
  $F_{i}$ queries can be evaluated in polynomial time.
\begin{proof}
    Let $G$ be the given data graph, $f_{i}^{1}$ be the query to be
    evaluated and $\val$ be the given valuation for $\fv(f_{i}^{1})$.
    For every pair $\struct{v_{1}, v_{2}}$ of nodes in $G$ and every
    sub-expression $e_{i}^{1}$ of $f_{i}^{1}$,
    check if $\struct{v_{1}, v_{2}} \in e_{i}^{1}[\val](G)$ by calling the
    oracle. Draw an edge labeled $e_{i}^{1}$ from $v_{1}$ to $v_{2}$
    iff the oracle answers positively. Call the resulting data graph
    $G'$.

  Perform the standard product construction of $\aut{f_{i}^{1}}$ with
  $G'$ (this can be done since $G'$ also treats sub-expressions in
  $E_{i}$ as a single letter). A pair $\struct{v,v'}$ belongs to
  $f_{i}^{1}[\val](G)$ iff $(v',q_{f})$ is reachable from $(v,q_{0})$
  in the product system, where $q_{f}$ and $q_{0}$ are some final and
  initial states of $\aut{f_{i}^{1}}$ respectively.

  For the case of $F_{0}$, the only sub-expressions that can not be
  handled directly by standard automata are those of the form
  $a[c]$. Given the evaluation $\val$, such expressions can be evaluated
  in linear time. Hence, in this case, the above procedure takes
  polynomial time without any oracle.
\end{proof}

\Repeat{Theorem}{thm:complexityEvaluation}
  For queries in $E_{i}$, the evaluation problem belongs to $\S_{i}$.
\begin{proof}
  By induction on $i$. For the base case $i=1$, let $e_{1}^{1}$ be the
  given expression and let $\val$ be the given valuation for the free
  variables of $e_{1}^{1}$. We begin by non-deterministically choosing
  one of the sub-expressions for every sub-expression
  $e_{1}^{2}+e_{1}^{3}$ of $e_{1}^{1}$. This will result in an
  expression $u_{1}^{1}$ in $U_{1}$.
  Now, to check if $\struct{v,v'} \in u_{1}^{1}[\val](G)$, we proceed by
  recursion on the structure of $u_{1}^{1}$ as follows.
  \begin{itemize}
      \item To check if $\struct{v_{1},v_{2}} \in (u_{1}^{2}\cdot
      u_{1}^{3})[\val](G)$, we non-deterministically guess a node
      $v_{3}$ and recursively check that $\struct{v_{1},v_{3}} \in
      u_{1}^{2}[\val](G)$ and $\struct{v_{3},v_{2}} \in u_{1}^{3}[\val](G)$.
  \item To check if $\struct{v_{1},v_{2}} \in a
      \downarrow_{x}(u_{1}^{2})[\val](G)$, we non-deterministically
      choose an $a$-successor $v_{3}$ of $v_{1}$ and note the data
      value $d$ of the $a$-labeled edge from $v_{1}$ to $v_{3}$. Next
      we recursively check that $\struct{v_{3},v_{2}} \in u_{1}^{2}[\val[x \to
      d]](G)$.
  \item To check if $\struct{v_{1}, v_{2}} \in f_{0}^{1}[\val](G)$ for some
      expression $f_{0}^{1}$ in $F_{0}$, we proceed as in the proof of
      Lemma~\ref{lem:evalFiWithOracleForEi}.
  \end{itemize}

  Next we inductively assume that evaluating expressions in
  $E_{i}$ is in $\S_{i}$. To evaluate expressions in $E_{i+1}$, we
  proceed in the same way as in the base case. The only difference is
  in the case where we have to check $\struct{v_{1}, v_{2}} \in
  f_{i}^{1}[\val](G)$ for some expression $f_{i}^{1}$ in $F_{i}$.
  From Lemma~\ref{lem:evalFiWithOracleForEi}, this can be done in
  polynomial time with an oracle for evaluating expressions in
  $E_{i}$. Since, by induction hypothesis, the oracle itself is in
  $\S_{i}$, we conclude that evaluating expressions in $E_{i+1}$ is in
  $\S_{i+1}$.
\end{proof}

\newtheorem*{lemma2}{Lemma B.2}
\begin{lemma2}
    \label{lem:shortWitness}
    Suppose $e_{i}^{1}$ is an expression in $E_{i}$, with $|\aut{e}|
    \le k$ for every sub-expression $e$ of $e_{i}^{1}$. Let $\val$ be
    a valuation for $\fv(e_{i}^{1})$. If there is a data path in
    $L(e_{i}^{1}, \val)$ connecting $v_{1}$ to $v_{2}$ in a data graph
    with $n$ nodes, then there is such a data path of length at most
    $(k^{2} n)^{i}$.
\end{lemma2}
\begin{proof}
    By induction on $i$. For the base case $i=1$, we begin by giving
    short witnesses for subexpressions of $e_{1}^{1}$. For a
    subexpression $f_{0}^{1}$, the automaton $\aut{f_{0}^{1}}$ will have at most
    $k$ states. Since the valuation does not change, we can infer
    from standard pumping arguments that if a data path in the
    language of $f_{0}^{1}$ connects $v_{1}$ to $v_{2}$, there is such
    a data path of length at most $kn$. Next we consider $e_{1}^{1}$. For
    every data path $w$ in the language of $e_{1}^{1}$, we infer from
    Proposition~\ref{prop:UnionNormalForm} that there is an expression
    $u_{1}^{1}$ that contains $w$ in its language. The path $w$ may be
    split into sub-paths, each of which is in the language of some
    sub-expression of $u_{1}^{1}$, of the form $f_{1}^{1}$ or $a
    \bind{x}$. Since $|\aut{u_{1}^{1}}| \le k$, the number of
    such sub-expressions, and hence the number of sub-paths in $w$, is
    at most $k$. We have already seen that each sub-path can be
    replaced by one of length at most $kn$.  Hence, the total length
    of the path is at most $k^{2}n$.

    The induction step is similar, contributing a multiplicative
    factor of $k^{2}n$. Hence the result follows.
\end{proof}

\Repeat{Theorem}{thm:evalLevelInWP}
    Evaluating REWB queries in $E_{1}$, parameterized by the size of
    the query is in \wop{}.
\begin{proof}
    We will use \cite[Lemma~7, Theorem~8]{CFG2003}, which give machine
    characterizations for problems in \wop{}. They prove that a
    parameterized problem is in \wop{} iff there is a
    non-deterministic Turing machine that takes an instance $(x,k)$
    and decides the answer within $f(k)|x|^{c}$ steps, of which at
    most $f(k) \log |x|$ are non-deterministic (for some
    computable function $f$ and a constant $c$). Such a Turing machine
    exists for evaluating REWB queries in $E_{1}$. In such queries,
    every binding in the query is performed at most once in a path
    (since bindings are not iterable). Hence, the machine can first
    non-deterministically choose the data values for each binding in
    the query in the allowed number of non-deterministic steps. Then
    the expression can be treated as a standard regular expression, by
    substituting the guessed data values for the bindings.  The set of
    data values found in the data graph can be considered as a finite
    alphabet and evaluation can be done in polynomial time using
    standard automata theoretic techniques.
\end{proof}

Just like we get \wop{} from \wsat{} by replacing formulas with
circuits, we get \awp{} from \awsat{} by replacing formulas with
circuits. It has been proved in \cite[Theorem~17]{CFG2003} that a
parameterized problem is in \awp{} iff there is an alternating Turing
machine that takes an instance $(x,k)$ and decides the answer within
$f(k)|x|^{c}$ steps, of which at most $f(k) \log |x|$ are existential
or universal (let us call such machines \awp{} machines). We have seen
in Theorem~\ref{thm:evalLevelInWP} that evaluating REWB queries in
$E_{1}$ can be done by non-deterministic Turing machines with bounded
non-determinism (let us call them \wop{} machines). As we did in
Theorem~\ref{thm:complexityEvaluation}, we can evaluate REWB
queries in $E_{i}$ using an oracle hierarchy of height $i$, consisting of \wop{}
machines. In complexity theory, an oracle hierarchy of \NP{} machines
is known to be equivalent to an alternating Turing machine. It is
tempting to draw an analogous conclusion in parameterized complexity
theory, saying that an oracle hierarchy of \wop{} machines is
equivalent to an \awp{} machine. However, we have not been able to
prove such an equivalence for the following reason. In order to
simulate oracle calls in an alternating machine, one generally needs
as many non-deterministic steps as the number of calls to the oracle.
In the oracle hierarchy of \wop{} machines, the number of calls to an
oracle may be polynomial in the size of the input, but the number of
non-deterministic steps allowed in \awp{} machines is logarithmic in
the size of the input. We refer the interested reader to
\cite[Section 4]{CFG2003} for some discussions on how some results in
complexity theory fail in parameterized complexity theory.

For the query evaluation problem, we do not have upper bounds in
parameterized alternating time bounded classes. However, we can get an
upper bound in \emph{uniform}-\xnl{}, a parameterized space bounded class.

\smallskip
\Repeat{Theorem}{thm:xnlUpperBound}
    Evaluating REWB queries, with size of the query as parameter, is
    in uniform-\xnl{}.
\begin{proof}
    We give a space bounded non-deterministic algorithm. Suppose $n$
    is the size of the data graph and $k$ is the size of the
    expression and a pair of nodes is connected by a data path in the
    language of the expression. We know from
    Lemma~B.2 that there is such a data path of
    length at most $( (g(k))^{2}n)^{k}$, where $g(k)$ is an upper
    bound on $|\aut{e}|$ for any REWB $e$ of size $k$. A
    non-deterministic algorithm can guess and verify such a data path.
    It would have to store a counter to keep track of the length of
    the path, a valuation for variables in the expression and a node
    of the graph. All this needs space at most $\mathcal{O}(
    (g(k))^{2} \log n)$.
\end{proof}

\subsection{Lower Bounds}

\Repeat{Lemma}{lem:BoolFormulaSchema}
  Let $\phi$ be a Boolean formula over the propositional atoms
  $\prvar_{1}, \ldots, \prvar_{n}$ and $\val: \set{x_{1}, \ldots,
  x_{k}} \to \set{\prvar_{1}, \ldots, \prvar_{n}, *}$ be a valuation.
  The source of $G_{\phi}$ is connected to its sink by a data path
  in $L(e_{\eval}[k], \val)$ iff $\phi$ is satisfied by the truth
  assignment that sets exactly the propositions in $\set{\prvar_{1},
  \ldots, \prvar_{n}} \cap \range(\val)$ to true.
\begin{proof}
    By induction on the structure of the Boolean formula. Suppose
    $\phi$ is a positively occurring atom $\prvar_{j}$, satisfied by
    the truth assignment. Hence, the data value $\prvar_{j}$ is in
    $\set{\val(x_{1}), \ldots, \val(x_{k})}$. So the data path $\cdot
    \xrightarrow{\binom{a}{\pol}} \cdot \xrightarrow{\binom{a}{\nel}}
    \cdot \xrightarrow{\binom{b}{*}} \cdot
    \xrightarrow{\binom{\ponet}{\pol}} \cdot
    \xrightarrow{\binom{\prvarl}{\prvar_{j}}} \cdot
    \xrightarrow{\binom{e}{*}} \cdot$ is in $L(a\bind{x_{\pol}}
    a\bind{x_{\nel}}b\cdot \ponet[x_{\pol}^{=}]\cdot \prvarl[x_{1}^{=}
    \lor \cdots \lor x_{k}^{=}]e, \val)$, giving the desired data
    path in $G_{\phi}$.
    Conversely, suppose that the source of $G_{\phi}$ is connected to
    its sink by a data path in $L(e_{\eval}[k], \val)$. Since
    $e_{\eval}[k]$ begins with $a\bind{x_{\pol}} a \bind{x_{\nel}}$ and
    $G_{\phi}$ begins with $\cdot \xrightarrow{\binom{a}{\pol}} \cdot
    \xrightarrow{\binom{a}{\nel}} \cdot$, $x_{\pol}$, $x_{\nel}$ will
    have the values $\pol$, $\nel$ respectively. To reach the sink of
    $G_{\phi}$, the data path $\cdot \xrightarrow{\binom{b}{*}} \cdot
    \xrightarrow{\binom{\ponet}{\pol}} \cdot
    \xrightarrow{\binom{\prvarl}{\prvar_{j}}} \cdot
    \xrightarrow{\binom{e}{*}} \cdot$ has to be in
    $L(b\cdot\ponet[x_{\pol}^{=}]\cdot \prvarl[x_{1}^{=} \lor \cdots
    \lor x_{k}^{=}]e, \val)$. This implies that the data value
    $\prvar_{j}$ is in $\set{\val(x_{1}), \ldots, \val(x_{k})}$, which
    in turn implies that $\phi$ is satisfied by the truth assignment.  The
    argument is similar for a negatively occurring propositional
    atom. The induction steps are standard arguments based on the
    semantics of Boolean formulas.
\end{proof}

\Repeat{Theorem}{thm:evalLevel1NPHard}
  For queries in $E_{1}$, the evaluation problem is \NP{}-hard.
\begin{proof}
  We will reduce the satisfiability problem for Boolean formulas to
  the query evaluation problem. Suppose $\phi$ is a Boolean formula
  over the propositional atoms $\prvar_{1}, \ldots,
  \prvar_{n}$. The data graph is as follows.
  
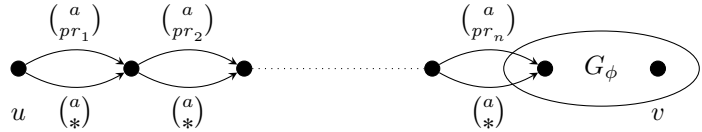
\begin{center}
    \begin{tikzpicture}[>=stealth, label distance=0.3\ml]
      \node[state, label=-90:$u$] (v0) at (0\ml,0\ml) {};
      \node[state] (v1) at ([xshift=1.5\ml]v0) {};
      \node[state] (v2) at ([xshift=1.5\ml]v1) {};
      \node[state] (v3) at ([xshift=2.5\ml]v2) {};
      \node[state] (v4) at ([xshift=1.5\ml]v3) {};
      \node[state, label=-90:$v$] (v5) at ([xshift=1.5\ml]v4) {};
      \node at (barycentric cs:v4=1,v5=1) {$G_{\phi}$};

      \draw[->] (v0) edge[bend left] node[pos=0.5, auto=left] {$\binom{a}{\prvar_{1}}$} (v1);
      \draw[->] (v0) edge[bend right] node[pos=0.5, auto=right] {$\binom{a}{*}$} (v1);
      \draw[->] (v1) edge[bend left] node[pos=0.5, auto=left] {$\binom{a}{\prvar_{2}}$} (v2);
      \draw[->] (v1) edge[bend right] node[pos=0.5, auto=right] {$\binom{a}{*}$} (v2);
      \draw[dotted] (v2)-- (v3);
      \draw[->] (v3) edge[bend left] node[pos=0.5, auto=left] {$\binom{a}{\prvar_{n}}$} (v4);
      \draw[->] (v3) edge[bend right] node[pos=0.5, auto=right] {$\binom{a}{*}$} (v4);

      \draw (barycentric cs:v4=1,v5=1) ellipse [x radius=1.3\ml, y radius=0.5\ml];
    \end{tikzpicture}
  \end{center}
  The ellipse at the end denotes the data graph $G_{\phi}$
  corresponding to the Boolean formula $\phi$, along with its source
  and and sink nodes. The query to be evaluated on this is $a
  \bind{x_{1}} a \bind{x_{2}} \cdots a \bind{x_{n}}
  e_{\eval}[n]$. To avoid too many parenthesis, we have not
  shown the scope of bindings. The scope of every binding extends till
  the end of the expression. We claim that the pair $\struct{u,v}$ is in the
  result of the query iff $\phi$ is satisfiable. Indeed, suppose
  $\struct{u,v}$ is in the result of the query. The data path from $u$ to
  $v$ will have two parts. The first one in $L(a \bind{x_{1}} a
  \bind{x_{2}} \cdots a \bind{x_{n}})$ from $u$ to the source of
  $G_{\phi}$, resulting in a valuation $\val: \set{x_{1}, \ldots,
  x_{n}} \to \set{\prvar_{1}, \ldots, \prvar_{n}, *}$. The second part is
  in $L(e_{\eval}[n], \val)$, connecting the source of $G_{\phi}$ to
  its sink.  From Lemma~\ref{lem:BoolFormulaSchema}, $\phi$ is
  satisfied by the truth assignment that sets $\prvar_{j}$ to true iff
  $\val(x_{j}) = \prvar_{j}$. Conversely, suppose $\phi$ is satisfied
  by some truth assignment $\a: \set{\prvar_{1}, \ldots, \prvar_{n}}
  \to \set{\mathrm{true}, \mathrm{false}}$. Consider the data path
  from $u$ to the source of $G_{\phi}$ that takes the edge labeled
  $\binom{a}{\prvar_{j}}$ if $\a(\prvar_{j}) = \mathrm{true}$ and
  takes the edge labeled $\binom{a}{*}$ otherwise.  This path is in
  $L(a \bind{x_{1}} a \bind{x_{2}} \cdots a \bind{x_{n}})$ and results
  in a valuation $\val$ such that $\mathrm{Range}(\val) \cap
  \set{\prvar_{1}, \ldots, \prvar_{n}}$ is precisely the set of
  propositional atoms set to true by the truth assignment $\a$. Since
  this truth assignment satisfies $\phi$, we conclude from
  Lemma~\ref{lem:BoolFormulaSchema} that the path can be continued
  from the source of $G_{\phi}$ to its sink.
\end{proof}

The gadgets we present before Lemma~\ref{lem:evalLevel1WSatHard} in
the main paper build on earlier ideas and bring out the difference in
the roles played by the data graph and the query, when reducing
satisfiability to query evaluation. We begin with an observation about
the REWB $e_{\eval}[k]$.

\newtheorem*{definition3}{Definition B.3}
\begin{definition3}[Indistinguishable variables]
    \label{def:indistinguishableVars}
    The variables $x_{1}, \ldots, x_{k}$ are said to be
    indistinguishable in an REWB $e$ if they are free in $e$ and for
    every condition $c$ appearing in $e$, for every data value $d$ and
    every pair of valuations $\val$ and $\val'$ with
    $\set{\val(x_{1}), \ldots, \val(x_{k})} = \set{\val'(x_{1}),
    \ldots, \val'(x_{k})}$, we have $d,\val \models c$ iff $d,\val'
    \models c$.
\end{definition3}
The variables $x_{1}, \ldots, x_{k}$ are indistinguishable in
$e_{\eval}[k]$. The intuition is that $e_{\eval}[k]$ treats the set
$\set{\val(x_{1}), \ldots, \val(x_{k})}$ as the set of propositional
atoms that are set to true. Any valuation $\val'$ with
$\set{\val'(x_{1}), \ldots, \val'(x_{k})} = \set{\val(x_{1}), \ldots,
\val(x_{k})}$ will have the same meaning, as far as $e_{\eval}[k]$ is
concerned.

Suppose $\prvars=\set{\prvar_{1}, \ldots, \prvar_{n}}$ is a set of
propositional atoms and $x_{1}, \ldots, x_{k}$ are
variables indistinguishable in some REWB $e$. We would like to check
if the source of some graph $G$ is connected to its sink by a data path
in the language of $e$, for some injective valuation $\val:
\set{x_{1}, \ldots, x_{k}} \to \prvars$.
The data graph $G[\exists k/\prvars]\circ G$ and the expression
$e[\exists k]\circ e$ defined in the main paper have been designed to
achieve this.

Suppose $\val$ is a valuation of some variables, whose domain does not
intersect with $\set{x_{1}, \ldots, x_{k}}$. We denote by
$\val[\set{x_{1}, \ldots, x_{k}} \xra{1:1} \prvars]$ the set of
valuations $\val'$ that extend $\val$ such that
$\mathrm{domain}(\val') = \mathrm{domain}(\val) \cup \set{x_{1},
\ldots, x_{k}}$, $\val'$ is injective on $\set{x_{1}, \ldots, x_{k}}$
and $\set{\val'(x_{1}), \ldots, \val'(x_{k})} \subseteq \prvars$.
\newtheorem*{lemma4}{Lemma B.4}
\begin{lemma4}
    \label{lem:existsGadget}
    Suppose $x_{1}, \ldots, x_{k}$ are indistinguishable in the REWB
    $e$ and $\val$ is a valuation for $\fv(e) \setminus \set{x_{1},
    \ldots, x_{k}}$. The source of $G[\exists k/\prvars]\circ G$ is
    connected to its sink by a data path in $L(e[\exists k]\circ e,
    \val)$ iff there exists $\val' \in \val[\set{x_{1}, \ldots, x_{k}}
    \xra{1:1} \prvars]$ and a data path in $L(e, \val')$ connecting the
    source of $G$ to its sink.
\end{lemma4}
\begin{proof}
    Suppose the source of $G[\exists k/\prvars]\circ G$ is connected
    to its sink by a data path in $L(e[\exists k]\circ e,\val)$. When
    this path reaches the source of $G$, the updated valuation $\val'$
    is in $\val[\set{x_{1}, \ldots, x_{k}}
    \xra{1:1} \prvars]$. Hence, the continuation of the path from
    the source of $G$ to its sink is in $L(e,\val')$.
    
    Conversely, suppose there exists $\val' \in \val[\set{x_{1},
    \ldots, x_{k}} \xra{1:1} \prvars]$ and the source of $G$ is
    connected to its sink by a data path in $L(e, \val')$. There is a
    data path $w_{1}$ from the source of $G[\exists k/\prvars]\circ G$
    to the source of $G$ in $L(a_{1}^{*} a_{1} \bind{x_{1}} a_{1}^{*}
    a_{1}\bind{x_{2}} a_{1}^{*} \cdots a_{1}^{*} a_{1}\bind{x_{k}}
    a_{1}^{*})$, resulting in a valuation $\val'' \in \val[\set{x_{1},
    \ldots, x_{k}} \xra{1:1} \prvars]$ such that $\set{\val''(x_{1}),
    \ldots, \val''(x_{k})} = \set{\val'(x_{1}), \ldots,
    \val'(x_{k})}$. Since $x_{1}, \ldots, x_{k}$ are indistinguishable
    in $e$, we infer that the source of $G$ is connected to its sink
    by a data path $w_{2}$ in $L(e,\val'')$.  The two data paths
    $w_{1}$ and $w_{2}$ can be concatenated to get a data path in
    $L(e[\exists k]\circ e,\val)$ connecting the source of $G[\exists
    k/\prvars]\circ G$ to its sink.
\end{proof}

\Repeat{Lemma}{lem:evalLevel1WSatHard}
    Let $\phi$ be a Boolean formula over the set $\prvars$ of
    propositions and $k \in \Nat$. We can construct in
    polynomial time a data graph $G$ and an REWB $e_{1}^{1}$
    satisfying the following conditions.
    \begin{enumerate}
        \item The source of $G$ is connected to its sink by a data
            path in $L(e_{1}^{1})$ iff $\phi$ has a satisfying
            assignment of weight $k$.
        \item The size of $e_{1}^{1}$ depends only on $k$.
    \end{enumerate}
\begin{proof}
    The required data graph is $G[\exists k/\prvars]\circ G_{\phi}$
    and $e_{1}^{1}$ is $e[\exists k]\circ e_{\eval}[k]$. The
    correctness follows from Lemma~B.4 and
    Lemma~\ref{lem:BoolFormulaSchema}.
\end{proof}

\Repeat{Lemma}{lem:universalGadgetInduction}
    Let $i \in \set{1,\ldots, k}$ and $\val_{i}$ be a valuation for
    $\fv(e^{i}) \setminus \set{x_{1}, \ldots, x_{i}}$. The source of
    $G_{i}$ is connected to its sink by a data path in $L(e^{i},
    \val_{i})$ iff for every $\val \in \val_{i}[\set{x_{1}, \ldots,
    x_{i}} \to \prvars]$, there is a data path in $L(e^{0},\val)$
    connecting the source of $G_{0}$ to its sink.
\begin{proof}
    By induction on $i$. For the base case, we have $e^{1} =
    b_{1}(a_{1}\bind{x_{1}}(e^{0}a_{1}[x_{1}^{=}])c_{1})^{*}$. Suppose
    there is a data path $w \in L(e^{1}, \val_{1})$ connecting the
    source of $G_{1}$ to its sink. We see from
    Figure~\ref{fig:UniversalGadget} that $w$ has to start from the
    edge labeled $b_{1}$, followed by $\binom{a_{1}}{\prvar_{1}}$,
    assigning $\prvar_{1}$ to $x_{1}$. Then $w$ has to go from the
    source of $G_{0}$ to its sink using a sub-path in $L(e^{0},
    \val_{1}[x_{1} \to \prvar_{1}])$.  From the sink of $G_{0}$, $w$
    is forced to take the edge labeled $\binom{a_{1}}{\prvar_{1}}$, in
    order to satisfy the condition $[x_{1}^{=}]$. The next letter of
    $w$ is $c_{1}$, which leads to a node from where the only outgoing
    edge is labeled with $\binom{a_{1}}{\prvar_{2}}$. This forces $w$
    to have a sub-path from the source to the sink of $G_{0}$ in
    $L(e^{0}, \val_{1}[x_{1} \to \prvar_{2}])$. We can similarly infer
    that for every $j \in \set{1,\ldots, n}$, $w$ has sub-paths  in
    $L(e^{0},\val_{1}[x_{1} \to \prvar_{j}])$ connecting the source of
    $G_{0}$ to its sink.

    Conversely, suppose that for every $j \in \set{1, \ldots, n}$, there is
    a data path $w_{j} \in L(e^{0},\val_{1}[x_{1} \to \prvar_{j}])$
    connecting the source of $G_{0}$ to its sink. The data path $b_{1} (
    \binom{a_{1}}{\prvar_{j}}w_{j}\binom{a_{1}}{\prvar_{j}}c_{1})_{j\in
    \set{1, \ldots, n}} \in L(e^{1}, \val_{1})$ connects the source
    of $G_{1}$ to its sink.

    The induction step is similar to the base case.
\end{proof}

\Repeat{Lemma}{lem:universalGadgetCoversAll}
    Let $\val$ be a valuation for $\fv(e) \setminus \set{x_{1},
    \ldots, x_{k}}$ for some REWB $e$. The source of $G[\forall
    k/\prvars]\circ G$ is connected to its sink by a data path in
    $L(e[\forall k]\circ e, \val)$ iff for all
    $\val' \in \val[\set{x_{1}, \ldots, x_{k}} \xra{1:1} \prvars]$,
    the source of $G$ is connected to its sink by a data path in
    $L(e,\val')$.
\begin{proof}
    Suppose the source of $G[\forall
    k/\prvars]\circ G$ is connected to its sink by a data path in
    $L(e[\forall k]\circ e, \val)$. We infer
    from Lemma~\ref{lem:universalGadgetInduction} that for all
    $\val' \in \val[\set{x_{1}, \ldots, x_{k}} \to \prvars]$, there is
    a data path $w_{\val'}$ from the source of $G_{0}$ to its sink in
    $L(e^{0},\val')$. For any such $\val'$ that is injective on
    $\set{x_{1}, \ldots, x_{k}}$, $w_{\val'}$ can not have $\sk$
    edges. The reason is that $e^{0}$ enforces $\val'(x_{i}) =
    \val'(x_{j})$ for some distinct $i,j \in \set{1,\ldots, k}$ in
    order to take a $\sk$ edge, but this is not possible for $\val'$
    if it is injective on $\set{x_{1}, \ldots, x_{k}}$. Hence, for
    extensions $\val'$ that are injective on $\set{x_{1}, \ldots,
    x_{k}}$, the source of $G$ is connected to its sink by a data path
    in $L(e,\val')$.

    Conversely, suppose that for all $\val' \in \val[\set{x_{1},
    \ldots, x_{k}} \xra{1:1} \prvars]$, the source of $G$
    is connected to its sink by a data path in $L(e,\val')$. This data path
    also connects the source of $G_{0}$ to its sink, and is in
    $L(e^{0}, \val')$. For valuations $\val' \in
    \val[\set{x_{1}, \ldots, x_{k}} \to \prvars]$ that are not injective
    on $\set{x_{1}, \ldots, x_{k}}$, there is a data path in $L(e^{0},
    \val')$ connecting the source of $G_{0}$ to its sink, which uses
    one of the $\sk$ edges from the source of $G_{0}$ to its sink.
    Hence, for all $\val' \in \val[\set{x_{1}, \ldots, x_{k}} \to
    \prvars]$, there is a data path in
    $L(e^{0},\val')$ connecting the source of $G_{0}$ to its sink. We
    conclude from Lemma~\ref{lem:universalGadgetInduction} that the
    source of $G[\forall k/\prvars]\circ G$ is connected to its sink
    by a data path in $L(e[\forall k]\circ e, \val)$.
\end{proof}

\Repeat{Lemma}{lem:evalAWSatHard}
    Given an instance of the weighted quantified satisfiability
    problem, We can construct in polynomial time a data graph $G$ and an
    REWB $e_{1 + k_{2} + k_{4} +\cdots}^{1}$ satisfying the following conditions.
    \begin{enumerate}
        \item The source of $G$ is connected to its sink by a data
            path in $L(e_{1 + k_{2} + k_{4} +\cdots}^{1})$ iff the
            given instance of the weighted quantified satisfiability
            problem is a yes instance.
        \item The size of $e_{1 + k_{2} + k_{4} +\cdots}^{1}$ depends
            only on $k_{1}, \ldots, k_{\ell}$.
    \end{enumerate}
\begin{proof}
    The required data graph $G$ is $G[\exists k_{1}/\prvars_{1}] \circ
    G[\forall k_{2}/\prvars_{2}] \circ \cdots \circ G_{\phi}$ and the
    required REWB $e_{1 + k_{2} + k_{4} +\cdots}^{1}$ is $e[\exists
    k_{1}] \circ e[\forall k_{2}] \circ \cdots \circ e_{\eval}[k_{1} +
    \cdots + k_{\ell}]$. We assume that $\circ$ associates to the
    right, so $G_{1} \circ G_{2} \circ G_{3}$ is $G_{1} \circ (G_{2}
    \circ G_{3})$ and $e^{1} \circ e^{2} \circ e^{3}$ is $e^{1}
    \circ(e^{2} \circ e^{3})$. Suppose that the source of $G$ is
    connected to its sink by a data path in $L(e_{1 + k_{2} + k_{4}
    +\cdots}^{1})$.  Lemma~B.4 ensures that there
    exists an injective valuation $\val_{1}: \set{x_{1}, \ldots,
    x_{k_{1}}} \to \prvars_{1}$ such that the source of $G[\forall
    k_{2}/\prvars_{2}] \circ \cdots \circ G_{\phi}$ is connected to its
    sink by a data path in $L(e[\forall k_{2}] \circ \cdots \circ
    e_{\eval}[k_{1} + \cdots + k_{\ell}], \val_{1})$. Then
    Lemma~\ref{lem:universalGadgetCoversAll} ensures that for all
    $\val_{2} \in \val_{1}[\set{x_{k_{1}+1}, \ldots, x_{k_{1}+k_{2}}}
    \xra{1:1} \prvars_{2}]$, the source of $G[\exists k_{3}/\prvars_{3}]
    \circ \cdots \circ G_{\phi}$ is connected to its sink by a data path
    in $L(e[\exists k_{3}] \circ \cdots \circ e_{\eval}[k_{1} + \cdots
    k_{\ell}], \val_{2})$. This argument can be repeated $\ell$ times to
    cover all alternations in the Boolean quantifiers. Finally,
    Lemma~\ref{lem:BoolFormulaSchema} ensures that the truth assignments
    corresponding to the valuations all satisfy the Boolean formula
    $\phi$. The converse direction is similar.
\end{proof}

\end{document}